\documentclass[sigconf, authorversion]{acmart}
\usepackage{textcomp}
\DeclareUnicodeCharacter{2731}{\textasteriskcentered}
\usepackage{bbold}      % double-struck numbers
\usepackage{hyperref}
\usepackage{multirow}
\usepackage{bm}             % replaces \bm --> \bm
\usepackage{accents}           % 
\usepackage{amsmath}

\newcommand{\figref}[1]{Figure \ref{#1}}

\usepackage{pdfpages}       % For adding comments
\newboolean{showcomments} 
\setboolean{showcomments}{true} 
\newcommand{\slow}[1]{\ifthenelse{\boolean{showcomments}}
{ \textcolor{red}{(SL: #1)}}{}}
\newcommand{\yx}[1]{\ifthenelse{\boolean{showcomments}}
{ \textcolor{blue}{(YX: #1)}}{}}
\newcommand{\wc}[1]{\ifthenelse{\boolean{showcomments}}
{ \textcolor{brown}{(WC: #1)}}{}}
\newcommand{\adam}[1]{\ifthenelse{\boolean{showcomments}}
{ \textcolor{teal}{(AW: #1)}}{}}
\AtBeginDocument{%
  }

\setcopyright{acmlicensed}
\copyrightyear{2026}
\acmYear{2026}
\acmDOI{XXXXXXX.XXXXXXX}
\begin{document}

\title{Enhancing Data Center Low-Voltage Ride-Through
}


% todo:
\author{Yiheng Xie}
\authornote{Both authors contributed equally to this research.}
\email{yxie5@caltech.edu }
\affiliation{%
  \institution{California Institute of Technology}
  \city{Pasadena}
  \state{California}
  \country{USA}
}
\author{Wenqi Cui}
\authornotemark[1]
\email{wenqicui@nyu.edu}
\affiliation{%
  \institution{New York University}
  \city{Brooklyn}
  \state{New York}
  \country{USA}
}

\author{Adam Wierman}
\email{adamw@caltech.edu }
\affiliation{%
  \institution{California Institute of Technology}
  \city{Pasadena}
  \state{California}
  \country{USA}
}
% \affiliation{%

% \author{A Chan}
% \affiliation{%


\begin{abstract}
% Backgroud
Data center loads have expanded significantly in recent years. Compared to traditional loads, data centers are highly sensitive to voltage deviations and thus their protection mechanisms trip more proactively during voltage fluctuations. During a grid fault, simultaneous tripping of large-scale data centers can further destabilize the transmission system and even lead to cascading failures. In response, transmission system operators are imposing voltage ride-through (VRT) requirements for data centers. %However, the mechanisms to achieve low-voltage ride-through remain largely understudied in data centers.
% While extensive studies have been done on device-level fault ride-through capabilities in generators and inverter-based resources, fault ride-through capabilities and analysis of data centers have not been studied. 
%
In this work, we enhance the VRT capability of data centers by designing voltage controllers for their internal power distribution network. 
% In this work, we propose a voltage controller for data centers to stabilize the voltages within their internal power distribution system, which enables the facility to remain online during voltage disturbances in the external utility grid. 
We first systematically analyze VRT standards and the controllable resources related to data centers. These resources enable the design of voltage control strategies to regulate voltages internal to the data center, thereby allowing loads to remain online during voltage disturbances from the external transmission grid.
% The proposed voltage control strategies regulate voltages internal to the data center, allowing loads to remain online during voltage disturbances from the external transmission grid.
We study and contrast both centralized and decentralized controllers that unify the control of heterogeneous flexible resources. %Our networked approach unifies the control of heterogeneous flexible resources within a data center, including cooling systems, computing loads, UPS units, and utility-scale batteries. 
Additionally, we construct an integrated test system that simulates both the transient fault response of the transmission system and the data center distribution network. Case studies demonstrate that the proposed voltage control mechanisms provide effective yet simple solutions to enhance data center low-voltage ride-through capability.

% Data center loads are highly sensitive to voltage deviations. As such, data center protection mechanisms trip more proactively compared to traditional loads. During a grid event such as line fault or loss of generation, simultaneous tripping of large-scale data centers can further destabilize the transmission system, with real-world incidents already emerging. Combined with the global proliferation of large-scale data centers, the lack of fault ride-through capabilities is becoming a primary concern for transmission systems operators. While extensive studies have been done on device-level fault ride-through capabilities in generators and inverter-based resources, fault ride-through grid codes are only recently being proposed for large loads. 
% % This work
% In this work, we propose a data center voltage regulation scheme that stabilizes the voltages within the data center distribution system, so that the data center remains online during a voltage disturbance. We present and compare both centralized and decentralized control schemes which are suitable for different settings. This networked approach to voltage ride through unifies the control of heterogeneous flexible resources within a datacenter such as cooling, computing loads, UPS, and utility-scale batteries to jointly achieve voltage ride-through on a system level.
\end{abstract}

\begin{CCSXML}
<ccs2012>
 <concept>
  <concept_id>00000000.0000000.0000000</concept_id>
  <concept_desc>Do Not Use This Code, Generate the Correct Terms for Your Paper</concept_desc>
  <concept_significance>500</concept_significance>
 </concept>
 <concept>
  <concept_id>00000000.00000000.00000000</concept_id>
  <concept_desc>Do Not Use This Code, Generate the Correct Terms for Your Paper</concept_desc>
  <concept_significance>300</concept_significance>
 </concept>
 <concept>
  <concept_id>00000000.00000000.00000000</concept_id>
  <concept_desc>Do Not Use This Code, Generate the Correct Terms for Your Paper</concept_desc>
  <concept_significance>100</concept_significance>
 </concept>
 <concept>
  <concept_id>00000000.00000000.00000000</concept_id>
  <concept_desc>Do Not Use This Code, Generate the Correct Terms for Your Paper</concept_desc>
  <concept_significance>100</concept_significance>
 </concept>
</ccs2012>
\end{CCSXML}

\ccsdesc[500]{Hardware~Power networks}
\ccsdesc[500]{Hardware~Enterprise level and data centers power issues}
% \ccsdesc[300]{Theory of computation~Convex optimization}
\ccsdesc[500]{Hardware~Power estimation and optimization}
\ccsdesc[300]{Hardware~Smart grid}
\ccsdesc[300]{Computing methodologies~Modeling and simulation}
\keywords{Data centers, low-voltage ride-through, voltage control}

% todo:

\maketitle

\section{Introduction}
% \IEEEPARstart{D}{istribution}
% Data center loads are growing rapidly in the US, China, and many parts of Europe. 
Data center electricity demand is increasing at an unprecedented pace worldwide \cite{shehabi_united_2024, goldman_sachs_ai_2025}. Unlike traditional loads, the operations of data centers are more sensitive to fluctuations in voltage and frequency \cite{babu_chalamala_data_2025}. 
% Data centers differ from traditional loads in that their operations are more sensitive to fluctuations in voltage and frequency. 
As a result, during a grid event, data center uninterruptible power supply (UPS) systems are proactive in disconnecting from the grid and switching to local backup energy supplies. Therefore, from the grid perspective, data centers have a lower fault ride-through tolerance \cite{babu_chalamala_data_2025}.

Traditionally, voltage ride-through (VRT) grid codes are typically required for generation assets, while loads are treated as uncontrollable (or passive)~\cite{zeb_faults_2022, moheb_comprehensive_2022, howlader_comprehensive_2016}. As a result, data centers have not faced incentives or regulatory obligations to enhance their VRT capability. 
% There has typically been no incentive or requirement for data centers to improve VRT capabilities. 
However, as hundred-megawatt-scale data centers continue to be interconnected, the fault ride-through capabilities of these large loads raise serious concerns about the overall power system stability\cite{nerc_nerc_2025,patrick_gravois_ercot_2025}. Simultaneous disconnection of large loads during a grid fault event can lead to cascading failure.

For example, on July 10, 2024, a permanent fault on a 340kV transmission line in the Eastern Interconnection of the United States caused a series of 6 voltage violation events over 82 seconds, with each event's duration ranging from 42 to 66 milliseconds and magnitude ranging from 0.24 to 0.4 per unit. As a result, approximately 1.5 Gigawatt of voltage-sensitive load was lost due to demand-side protection schemes. 
Subsequent analysis revealed that the entire 1.5 Gigawatt load was associated with data centers.
% It was later determined that the entirety of the 1.5 Gigawatt of load was associated with data centers. 
The system voltage rose to 1.07 per unit, and emergency mitigation actions were conducted. A comprehensive incident report can be found in \cite{nerc_nerc_2025}.

In the near future, existing and new data centers will face more stringent voltage ride-through requirements. In fact, transmission system operators in France, Ireland, Denmark, and Texas are already proposing VRT grid codes for large loads \cite{european_union_commission_2016, rte_article_2024, eirgrid_voltage_2025, myen_technical_2024, patrick_gravois_ercot_2025}. Broadly speaking, on the generation side, tighter VRT ride-through requirements for inverter-based power electronics are already being imposed by system operators \cite{nerc_2024_2025}. Considering that power electronic devices account for the majority of data center loads, similar tightening of VRT requirements for data centers is anticipated. Consequently, enhancing VRT capability in both new and existing data centers represents a critical consideration for advancing the integration of data centers into the power grid.
% data center developers moving forward.
% This paper aims to address the following key question: 
% \textit{How can we enhance the VRT capabilities of data centers?}.
% This incurs a significant gap between the low-voltage tolerance of data center facilities (such as UPS) and the higher requirements for them to stay connected with the grid during low-voltage events. 
% However, for data center operators, 
% the VRT capabilities and the mechanisms to enhance it is not clear, and there is a lack of literatire address this issue. This paper seeks to address the following questions: What is the 
% How to enhance VRT capability for new and existing data centers is an important question for data center developers going forward.

In the broader picture, to make data centers more grid-friendly, a variety of approaches have been developed that enable them to provide services to the power grid, including demand response and frequency regulation \cite{wierman_opportunities_2014, ghatikar_demand_2012, vaidhynathan_vulcan_2025, filippi_how_2023, fu_assessments_2020, wang_frequency_2019}. However, the voltage ride-through problem is fundamentally different: it requires a fast response to voltage disturbances occurring over a timescale of just a few milliseconds to seconds. Consequently, new control methods and solutions are needed. Despite the rapid growth of computing-intensive data center loads, there has been little systematic analysis of the quantitative requirements and controllable resources in the context of data center low-voltage ride-through. Aside from the novelty of the problem, another part of the reason is that, while modern data centers can be large enough to significantly influence power system operations, there is still a lack of open-source test systems to evaluate their interactions.

%%%%%%%%%%%%%%%%%%%%%%%%%%%%%%%%%%%%%%%%%%%%%%%%%%%%%%%
\subsection{Contributions}

In this work, we leverage the highly structured internal distribution networks within data centers to enhance VRT capabilities. By properly managing controllable devices inside the facility, internal voltages can be maintained within safe operating limits despite voltage fluctuations at the point of interconnection to the utility grid. This is achieved through the design of voltage control laws that regulate the power output of these devices in response to voltage deviations and time-varying power consumption.  Thus, data centers can remain connected to the grid without tripping because of the drop in voltage outside their tolerance limits.
% Concretely, we propose to enable and enhance the low-voltage ride-through capability of data centers via the design of both centralized and decentralized voltage controllers for their internal distribution systems. 
% This approach allows data center facilities to remain online without tripping because of the drop of voltaged outside their tolerance range.
% during grid disturbances.

Specifically, we conduct the first systematic study of low-voltage ride-through standards related to data centers, along with an analysis of the controllable resources within these facilities. Our study considers the diversity in data center infrastructure and communication capabilities, and we compare and contrast different control solutions for various settings. 
% underscoring the need for tailored controller designs that are best suited to different settings. 
In particular, we develop a centralized voltage controller and demonstrate its dependence on system-level communication and delay. When communication quality or capabilities are limited, decentralized control laws are necessary, and we adopt controllers that adjust local actions based solely on local voltage deviations.
Through the networked approach, we unify the control of heterogeneous flexible resources within a data center, including cooling systems, computing loads, UPS units, and utility-scale batteries. 

To evaluate the proposed controllers, we build an open-source test system that simulates both the transient fault response of the transmission system and the data center distribution network. Case studies demonstrate that the proposed voltage control mechanisms provide simple yet effective solutions to enhance low-voltage ride-through in data centers.

% We consider disturbance events that are well-represented by the pseudo steady-state phasor signals, such as faults due to loss of generation. As such, we adopt the Linear DistFlow circuit model to represent the power distribution system inside a data center.

% % Our setting
% In addition to UPS, data center contains other sources of flexibility such as cooling loads, dynamic voltage and frequency scaling (DVFS) on computing loads, and utility-scale battery energy storage systems (BESS)  on-site. 
% This work offers a framework for integrating these heterogeneous controllable devices in a unified control framework. By dynamically adjusting the real and reactive power load and consumption, the voltage within the data center is stabilized.

In summary, the main contributions of this paper are:
\begin{itemize}
\item We identify and formalize the emerging challenges and opportunities arising from low-voltage ride-through capabilities of data centers. We summarize the gaps in low-voltage ride-through standards and conduct a systematic analysis that incorporates the unique features of a data center.
    % \item A systematic consideration of voltage ride-through incorporating both the power systems perspective and unique feature of a data center
\item We construct and contrast centralized and decentralized controllers that achieve low-voltage ride-through for data centers. This provides practical guidance for the analysis and enhancement of low-voltage ride-through capabilities in a variety of data center architectures.

\item We develop an integrated test system that simulates both the transient response of the transmission system to faults
and the data center distribution network under voltage control. The open-source implementation provides a useful tool for studying power system dynamics with data center impacts for the research community.
\end{itemize}

%%%%%%%%%%%%%%%%%%%%%%%%%%%%%%%%%%%%%%%%%%%%%%%%%%%%%%%%%%%%%
\subsection{Related Work}

Our work is closely related to topics on grid-friendly data centers, voltage ride-through, and power system voltage control. Specifically, VRT is one form of grid service that data centers can provide to the grid. While VRT is well studied for individual devices (typically generators), the topic is not well explored in data centers, which are large loads with networked structures. Moreover, the VRT capability of a network of heterogeneous devices provides more design freedom and controllability. From this perspective, the enhancement of VRT capability can leverage a rich body of work on power system voltage control, which we briefly survey in this section. 

\subsubsection{Grid-friendly Data Centers}
% Hyperscale requires grid-friendly data centers 
In recent years, the number and scale of data centers have dramatically increased \cite{shehabi_united_2024}. This is largely due to the rapid growth in machine learning model sizes such as large language models. Data centers in the tens or hundreds of megawatt scale are becoming common. The sheer scale of data center loads in many regions around the world means that the existing approach of treating data centers as traditional loads is no longer appropriate due to their out-sized impact on the power system.

In response, there is emerging research interest in designing and operating data centers to be a "good" participant in the power grid. Terms such as grid-aware, grid-friendly, grid-integrated, or grid-interactive data centers have been introduced. The efforts can be categorized by the type of grid services provided, such as demand response \cite{wierman_opportunities_2014, ghatikar_demand_2012, vaidhynathan_vulcan_2025, filippi_how_2023}, frequency regulation \cite{fu_assessments_2020, wang_frequency_2019, filippi_how_2023}, and power ramp rate limits \cite{choukse_power_2025}, carbon reduction \cite{bostandoost_lacs_2024, bian_cafe_2024, priya_energy-minimizing_2024, lin_adapting_2023}. From this perspective, voltage ride-through has so far been overlooked as a grid service that data centers can (or must) provide. 
Another categorization is based on the source of flexibility, such as workload scheduling \cite{bostandoost_lacs_2024, epri_dcflex_nodate, radovanovic_carbon-aware_2021, hall_carbon-aware_2024, liu_data_2013, savasci_pads_2024}, dynamic voltage and frequency scaling (DVFS) \cite{fu_assessments_2020, hou_fine-grained_2023, priya_energy-minimizing_2024,chen_voltage_2025}, cooling \& thermal storage \cite{zhu_energy-efficient_2025, cao_toward_2023, liu_state---art_2020}, UPS \cite{filippi_how_2023, fu_assessments_2020}, utility-scale batteries \cite{vaidhynathan_vulcan_2025}, back-up generation \cite{liu_data_2013}. In this work, we consider all the relevant sources of flexibility that can provide fast control capabilities required by the VRT time scale. Together, this work presents a new perspective into grid-friendliness by systematically analyzing the VRT requirements and fast-timescale flexibility resources in the data center.

\subsubsection{Voltage Ride-through}
% \yx{Comment: We do mention FRT in two places: here and later on in Section 2. I am keeping the discussions here light and high-level, while Section 2 is more precise and quantitative.} \adam{sounds good}
% FRT
Of the related works, voltage ride-through is severely under-studied in the context of data centers at the time of writing. From the grid perspective, data centers have been treated as traditional passive loads over which grid operators exert no control. From the data center perspective, there is no incentive for riding through a grid power disturbance and risking equipment damage. 

However, VRT requirements exist to make sure that large energy devices -- whether they are generators or large loads -- do not disconnect at the first sign of a disturbance (e.g. voltage dip). If too many devices suddenly trip offline during a fault, the problem on the grid can get worse instead of better. System voltages and frequencies can rapidly increase and exceed their upper limits. VRT requirements ensure that the power system load conditions do not change drastically during a fault event. As VRT grid codes are being proposed for data centers (Section \ref{sec:vrt_grid_codes}), more sophisticated VRT implementation needs to be considered. 

% Data center protection system: UPS and device-level VRT
The core of data center protection mechanisms is the UPS, which isolates the data center devices from the grid in the event of a grid fault. 
Device-level VRT design and modeling is well-studied for traditional generators and inverter-based resources \cite{zeb_faults_2022, moheb_comprehensive_2022, howlader_comprehensive_2016}. On a high level, the approach is to jointly design hardware and corresponding control laws to ensure the device stays within thermal limits during an external fault event. In this work, we focus on a \textit{networked} approach to improve VRT capability of the entire power distribution system within a data center via voltage control of heterogeneous flexibility resources in a data center.

\subsubsection{Voltage Control}
% Voltage control.
We focus on voltage deviations occurring at timescales on the order of seconds. Conventionally, voltage regulation is performed using mechanical devices such as tap-changing transformers or capacitor banks~\cite{turitsyn_options_2011,bolognani_distributed_2013}. However, these devices cannot be adjusted frequently and therefore are unsuitable for the fast dynamics of LVRT events. For fast-timescale voltage control, extensive research has explored the use of inverter-based resources (such as energy storage, solar panels, wind turbines), which can adjust their power output rapidly and repeatedly without adversely affecting their lifespan~\cite{chen_data-driven_2020, wu_smart_2018}. For the distribution grid with communication capabilities, voltage control is typically formulated as a centralized or distributed optimization problem to coordinate the real-time power outputs of controllable nodes~\cite{zhang_optimal_2015, yeh_robust_2022, qu_optimal_2020}. To eliminate the communication requirements, decentralized control laws have also been proposed, in which each node adjusts its reactive power as a feedback function of its local voltage deviation~\cite{farivar_equilibrium_2013, li_real-time_2014, zhang_local_2013, baker_network-cognizant_2017}. For classes of feedback functions such as linear controllers and certain monotone control laws, the results in~\cite{ zhu_fast_2016, qu_optimal_2020, cui_decentralized_2022,feng_bridging_2023} establish the convergence of voltages to the safe operating range. However, these methods typically only include the control of reactive power. In addition, their applicability to data centers remains unclear. In this paper, we extend this framework to data center control with both active and reactive power. In addition, we demonstrate its applicability and trade-offs for data centers
with diverse infrastructure and communication capabilities.

\section{Challenge and Opportunity}

%%%%%%%%%%%%%%%%%%%%%%%%%%%%%%%%%%%%%%%%%%%%%%
\subsection{Voltage Ride-through Requirements}
\label{sec:vrt}

\begin{figure*}[ht!]
    \centering
    \includegraphics[width=.8\linewidth]{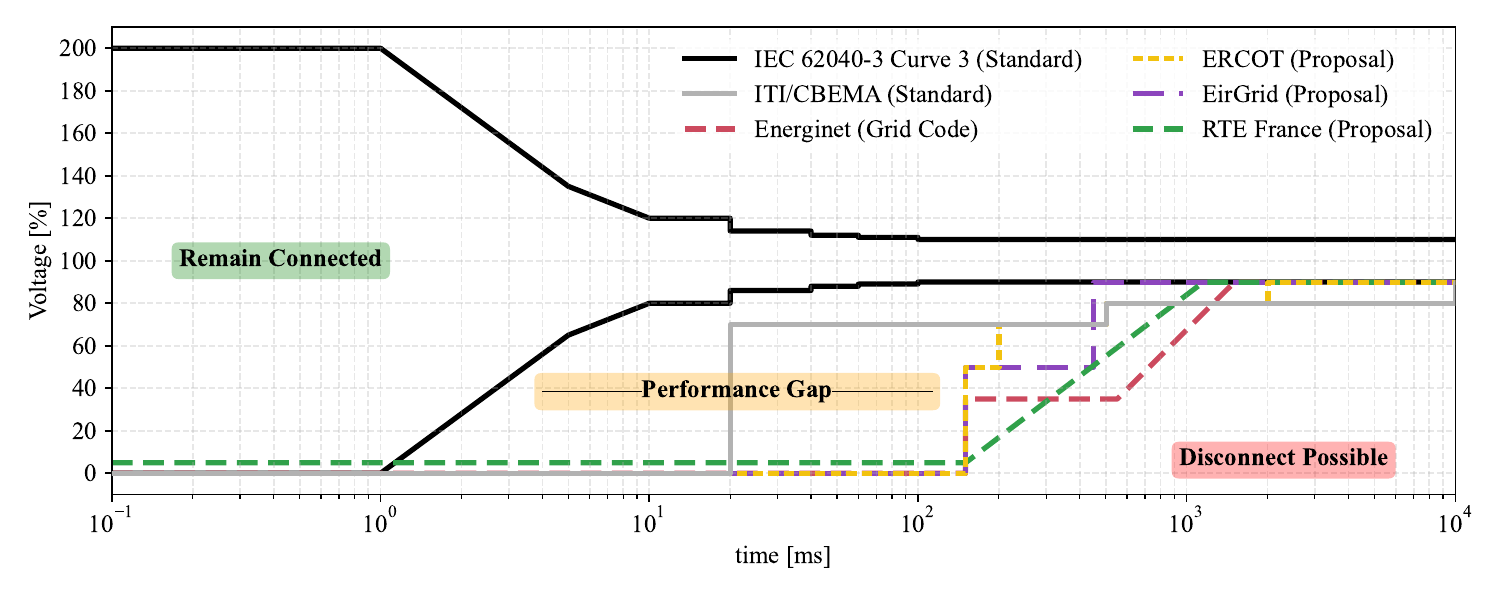}
    \caption{Voltage ride-through requirement curves.}
    \label{fig:VRT_curves}
\end{figure*}

%\subsubsection{Background}
With data centers becoming increasingly large and critical electrical loads, their impact on the stable operation of the power grid is becoming a major concern.  Unlike traditional loads, data center operations are highly sensitive to fluctuations in voltage and frequency. To protect their equipment, uninterruptible power supply (UPS) systems often disconnect data centers from the grid at the first sign of a disturbance, switching instead to local backup energy. While this ensures service continuity for the data center itself, it poses a serious challenge for the grid: when such massive loads disconnect simultaneously during a fault, they amplify the stress on the system rather than alleviating it. The sudden reduction in power demand can worsen instability and, in extreme cases, contribute to cascading failures. The July 2024 event in the Eastern Interconnection was a stark reminder of the dangers posed by the abrupt disconnection of large data center loads during grid disturbances.

To address this issue, power system operators are increasingly requiring data centers to provide fault ride-through (FRT) capabilities, which refers to the ability of grid-connected devices to remain connected and continue operating during a grid disturbance, commonly in the form of voltage or frequency deviations from nominal values. Specifically, voltage ride-through (VRT) refers to a device's ability to remain connected during a deviation in voltage magnitude for a short period of time without tripping offline. 
% Too many devices disconnecting from the grid during a fault can lead to further de-stabilization of the grid, potentially leading to cascading outages. The July 2024 instance in the Eastern Interconnection is a stark reminder of the dangers of the lack of VRT in large loads.
Requiring data centers to remain connected will ensure that the post-fault load conditions is similar to pre-fault conditions. Traditionally, VRT is required for bulk generators. However, due to the growing scale of data center loads, more stringent VRT requirements are being proposed for data centers recently. 

% Description of these curves.
A voltage ride-through curve specifies the duration and magnitude of a voltage deviation that a device must sustain. Several example voltage-ride-through specifications are shown in \figref{fig:VRT_curves}.
Particularly, the low-voltage ride-through (LVRT) requirement has seen numerous proposals, as low-voltage events are far more common. The LVRT performance increases (or the requirement is more stringent) as a trace sits farther right (longer duration) and farther down (larger magnitude). If a voltage deviation falls below and to the right of the LVRT curve, the facilities are permitted to trip offline; otherwise, they are required to remain online.

% A voltage deviation event below and to the right of the LVRT curve corresponds to possibly tripping offline, and a voltage deviation above and to the left of the LVRT curve corresponds to remaining online.

\subsubsection{Device Standards}

Device manufacturers for protection devices (e.g. UPS) and IT loads follow industry standards such as IEC 52040-3 \cite{iec_uninterruptible_2011} and the ITI/CBEMA\footnote{Information Technology Industry Council (ITI), formerly the Computer and Business Equipment Manufacturers Association (CBEMA)} curve \cite{noauthor_iti_nodate}. The IEC 62040-3 standard specifies VRT performance for UPS and the ITI/CBEMA curve specifies that for general IT loads. These specifications are the least strict as they are device standards adopted by device manufacturers rather than a grid code imposed by transmission system operators (TSOs). In other words, the device standards are developed from the load perspective, rather than the perspective of the power grid.

% ENTSO‑E
\subsubsection{Grid Codes}
\label{sec:vrt_grid_codes} 

In recent years, VRT grid codes are being proposed by TSOs in many regions of the world due to the growing impact of large loads such as data centers. The ENTSO‑E\footnote{European Network of Transmission System Operators for Electricity} Demand Connection Code (EU Regulation 2016/1388) provides an EU‑level framework for specifying fault ride-through capabilities for demand facilities (large loads). The the detailed LVRT envelope parameters are determined by national TSOs \cite{european_union_commission_2016}. Among national implementations, Energinet (Denmark) publishes one of the most stringent demand facility LVRT specifications \cite{myen_technical_2024}. Similar grid codes for large loads are being proposed by RTE (France)\footnote{Réseau de Transport d'Électricité, TSO for France} \cite{rte_article_2024} and EirGrid (Ireland) \cite{eirgrid_voltage_2025}. More specifically, EirGrid mentions a “stay‑connected + staged recovery” paradigm in its FRT study template and has publicly stated that a grid code modification to impose FRT on Large Energy Users (LEUs) is being developed \cite{eirgrid_voltage_2025, liam_ryan_shaping_2024}. In the U.S., ERCOT\footnote{Electric Reliability Council of Texas} is evaluating Large Electronic Load (LEL) ride‑through criteria with stepwise non‑trip regions \cite{patrick_gravois_ercot_2025}. Aside from Denmark, whose VRT requirement is in effect, the rest are proposals presently under review at the time of writing.

Figure \ref{fig:VRT_curves} shows that there is a substantial performance gap between the device standards and the proposed grid codes (note the logarithmic time scale). This presents a substantial challenge for existing and new data centers, which must meet grid requirements while also ensure adequate protection for voltage-sensitive computing devices. Fortunately, data centers typically consists of many types of flexible resources. With the appropriate controller, the performance gap can be closed.

%%%%%%%%%%%%%%%%%%%%%%%%%%%%%%%%%%%%%%%%%%%%%%%%%%%%%%%%%%%%%
\subsection{Controllable Devices within Data Centers}

\begin{figure*}[t]
    \centering
    \includegraphics[width=1.0\linewidth]{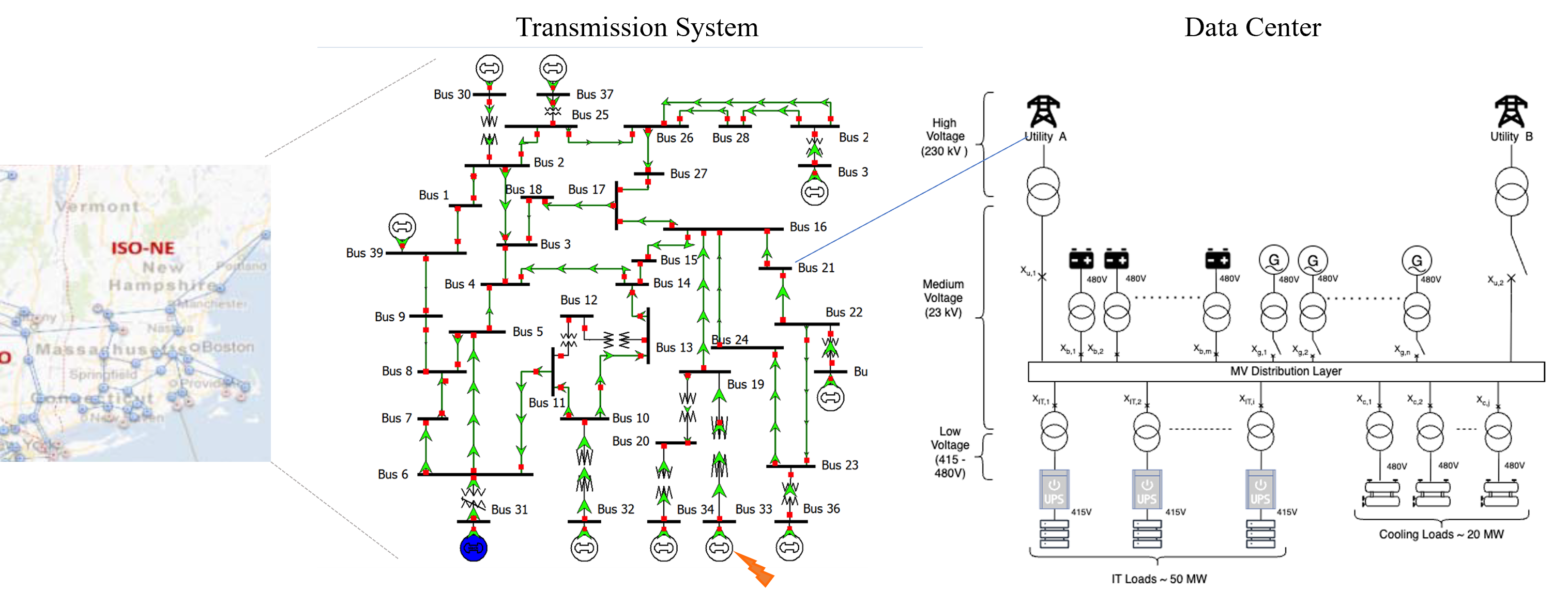}
    % {network_interconnection.png}
    \caption{Data center and grid interconnection. For demonstration purposes, the transmission system is IEEE New England 39 bus test system, and the data center network is adapted from the Vulcan Test Platform in~\cite{vaidhynathan_vulcan_2025}. }
    % \yx{@Wenqi: remove typo "high tension" in the Vulcan figure.}
    \label{fig:network_interconnection}
\end{figure*}

\begin{table*}[t]
    \centering
    \caption{Summary of Data Center Controllable Devices and Characteristics}
    \begin{tabular}{c|ccccc}
         \textbf{Controllable Device} & \textbf{\% of Total Capacity} & \textbf{Controllability} & \textbf{Response time} & \textbf{Investment Cost} & \textbf{Control Cost} \\
         \hline
         IT load & 50-70\% & real power (+/-) & milliseconds to seconds & low & high\\
         Cooling & 30-50\% & real power (+/-) & seconds to minutes & low & low\\
         Battery UPS & 50-70\% & real \& reactive power (+/-) & milliseconds & low & low \\
         Utility-scale BESS & varies & real \& reactive power (+/-) & milliseconds & high & low \\
         Backup Generation & 100\% & real power (+) & minutes & high & high
    \end{tabular}
    \label{tab:load_types}
\end{table*}

Unlike conventional generators, a data center is itself a networked system composed of internal power distribution infrastructure and interconnected devices. A representative structure of this internal grid is shown in \figref{fig:network_interconnection}~\cite{vaidhynathan_vulcan_2025}. By appropriately managing controllable devices within the facility, internal voltages can be kept within safe operating limits, even in the presence of fluctuations at the point of connection to the utility grid. Building on this observation, we propose to enhance the low-voltage ride-through capability of data centers through the design of voltage controllers for their internal distribution systems. In this section, we outline the opportunities provided by controllable devices inside data centers. The networked model and corresponding voltage control strategies are developed in the following sections.

On the devices side, a data center primarily consists of IT loads (computing, networking, and storage), cooling loads, battery UPS, and backup generators \cite{ahmed_review_2021}. Some data centers may additionally have centralized utility-scale battery energy storage systems (BESS). Their response time, power capacities, and types of control afforded are summarized in Table \ref{tab:load_types} 
\cite{rong_optimizing_2016, dayarathna_data_2016, sawyer_calculating_nodate}. 
For each device, the second column indicates the typical nameplate power capacity as a fraction of the total power import capacity at the utility interconnection. Note that this corresponds to the maximum possible power injection, not the average energy consumed. The exact quantities depend on a number of factors, including the type of computing jobs, redundancy requirements, cooling system design, environmental factors, and scale. 
The third column describes whether real and/or reactive power can be increased and/or decreased. The fourth column describes the typical response time of each flexibility resource. 
The last two columns characterize the upfront investment cost to enable controllability%
\footnote{If a device is already required and present in a data center, we do not consider that as a part of the investment cost here. The investment cost is the additional cost to enable controllability. For example, battery UPS is already present in most data centers so that the additional investment cost to enable control is low.} %
and the operating control cost associated with each control action.

\subsubsection{IT Loads} The computing loads often allow for dynamic voltage and frequency scaling (DVFS) to throttle the power usage of CPUs and GPUs at the hardware level \cite{priya_energy-minimizing_2024}. The response time is typically at the milliseconds level as the control is applied at the hardware level. The tradeoff is slower computation and longer residence time. DVFS has little to no up-front cost since it does not require additional hardware, but it potentially has a high control cost, since delaying latency-sensitive workloads can be very costly or impossible due to service-level agreements.

On the other hand, to up-modulate the power consumption (which can be useful in high-voltage ride-through), power padding can be achieved by injecting dummy computations. Computing loads can also be modulated at longer time scales with longer control delay via workload management and shutting down servers. However, these time scales (seconds to minutes) are not relevant to voltage ride-through.

\subsubsection{Cooling Loads \& Thermal Storage} The inherent thermal inertia of servers and buildings as well as dedicated thermal storage (e.g. chilled water tanks) also offer flexibility. Cooling loads can leverage this thermal inertia to temporarily reduce or increase power consumption \cite{cao_toward_2023, zhu_energy-efficient_2025, zhang_survey_2021}. The amount of flexibility depends on the total thermal inertia, which can be increased with on-site thermal storage\footnote{Thermal storage is generally much cheaper than electrical storage, but is more limited in controllability.}. Since all data centers requires a cooling system, and thermal inertia is an inherent physical property, allowing controllability in cooling system is a low-cost way to create additional flexibility. 

\subsubsection{Electrical Storage} UPS, utility scale BESS and flywheels are common examples of electro-chemical and electro-mechanical storage in data centers \cite{lechl_uncertainty-aware_2025, fu_assessments_2020}. These storage units interface with the grid via power electronics. Therefore, they offer fast and flexible control in both real and reactive power injection, so long as the control inputs are exposed by device manufacturers. Unlike the other flexibility resources, electrical storage units requires high capital cost but has relatively low control cost. In other words, they are expensive to build, but during operations, the relative cost of dispatching storage units is often much lower than the cost of delaying computing jobs.

In particular, UPS, unlike utility-scale BESS, are distributed across data center clusters, and therefore provides fine-grained controllability for different segments of the power distribution network. Moreover, UPS already have built-in energy storage elements available for dispatch. Thus the \textit{additional} investment cost to enable controllability may be low. In fact, a new class of grid-interactive UPS \cite{watson_data_2025, paananen_grid-interactive_2021, filippi_how_2023} are now being offered by hardware vendors, which provide additional programmable grid services in addition to traditional backup \& protection functionalities. 
% \footnote{A closely related class of devices is termed dynamic voltage restorer (DVR), which are series-connected protection devices with limited energy storage that reject short-duration voltage disturbances so that devices downstream are protected.}
We anticipate that the need for grid-interactive power supply and protection devices will grow, as large-scale data centers and the grid are operated closer towards their design limits.

\subsubsection{Backup Generation} Finally, on-site backup thermal generators such as diesel generators provide long-term backup energy supply, but take several minutes to start. Therefore, thermal backup generators are not relevant for the time scale of VRT.

%%%%%%%%%%%%%%%%%%%%%%%%%%%%%%%%%%%%%%%%%%%%%%
%%%%%%%%%%%%%%%%%%%%%%%%%%%%%%%%%%%%%%%%%%%%%%
\section{System Model}

We consider an interconnected transmission system and data center distribution system, shown in \figref{fig:network_interconnection}. In a power grid, the transmission system is the high-voltage network that delivers electricity over long distances from power plants to substations, where it is then stepped down for local distribution to consumers. The data center connects to one bus of the transmission network to draw power from the bulk power system.  
The data center network includes high-voltage interconnections to the transmission network, a medium-voltage distribution layer (typically several to tens of kV), and low-voltage (typically 480 V) connections supplying various loads. The distribution lines within the data center are relatively short, while the main supply line(s) connecting the data center to the transmission network may potentially be longer.

\subsection{Faults in the Transmission System}
% \wc{May move the section on low-voltage ride-through event here.}
The transmission system can experience faults during daily operations, such as loss of generation and line outages. A typical consequence of such faults is a sudden voltage drop across the network. Many faults can be cleared within 10 ms by the grid’s autonomous protection mechanisms. 
% Figure[] illustrates representative voltage dynamics, where the protection schemes restore the grid to normal operation within a few seconds.

The dynamics of power systems after a disturbance (including faults and changes of load) can be described by a set of DAEs as follows~\cite{xia_galerkin_2019, chiang_direct_nodate}:
% \yx{@Wenqi: citation isssue on chiang_direct_nodate here and below}

\begin{equation}\label{eq:Dynamic}
\left\{\begin{array}{c}
\dot{\bm{x}}=\bm{f}(\bm{x}, \bm{y}, \bm{a};\bm{p}^{DC},\bm{q}^{DC}) \\
\bm{0}=\bm{h}(\bm{x}, \bm{y}, \bm{a};\bm{p}^{DC},\bm{q}^{DC})
\end{array}\right.
\end{equation}
where $\bm{x} \in \mathbb{R}^{l}$, $\bm{y} \in \mathbb{R}^{m}$,  $\bm{a} \in \mathbb{R}^{d}$ are the state variables, algebraic variables and external input variables, respectively. 
The impact of data centers on the power system dynamics are reflected by their active power injection $\bm{p}^{DC}$ and reactive power injection $\bm{q}^{DC}$ at the point of connection with the transmission system. 
% If multiple data centers trip immediately following a fault, the direct effect is the switch of DAE to the system with $\bm{p}^{DC}=0$ and  $\bm{q}^{DC}=0$.
The differential equation  $\bm{f}: \mathbb{R}^{l} \times \mathbb{R}^{m} \times \mathbb{R}^{d} \rightarrow \mathbb{R}^{l}$ typically describes the internal dynamics of devices such as the speed and angle of generator rotors, the
dynamics transmission lines, dynamically modeled loads 
and their control systems. 
Correspondingly,  $\boldsymbol{x} \in \mathbb{R}^{l}$ is the state variables such as generator rotor angles, generator velocity, electromagnetic flux, and control system internal variables. The set of algebraic equations $\bm{h}: \mathbb{R}^{l} \times \mathbb{R}^{m} \times \mathbb{R}^{d} \rightarrow \mathbb{R}^{m}$ describes the electrical transmission system and interface equations.  Correspondingly, 
$\bm{y} \in\mathbb{R}^{m}$ is the algebraic variable such as voltage magnitude and angles. 
The external input variables $\bm{a}\in \mathbb{R}^{d}$ acting on the equations include power injection from generators, automatic generation control systems, fault-response actions, etc.~\cite{chiang_direct_nodate, cui_frequency_2024}.

In particular, the system voltages (a component of $\bm{y}$) changes as the DAEs evolve. Although these voltage fluctuations originate outside of data centers, they directly affect the internal voltage levels of data centers through the point of connection (as shown in Figure~\ref{fig:network_interconnection}). We will elaborate on this relation in the next subsection.

% Therefore, it is critical for data centers to remain connected and ride through the fault for a short duration, allowing the protection mechanisms to take effect.

%\adam{Do you want to introduce any notation for the transmission system here?}
%\yx{We don't need any notation since there is no theoretical analysis for the transmission system. We can briefly describe the DAE model used by ANDES. Should we? @Wenqi.}
%\wc{tried to add some descriptions, does it look reasonable?}

% The high-voltage transmission system captures the dynamics during a fault event external to the data center, while the data center distribution system circuit reflects the load response and internal power flow of a large-scale data center. 

%%%%%%%%%%%%%%%%%%%%%%%%%%%%%%%%%%%%%%%%%%%%%%%%%
\subsection{Data Center Model}

A data center interacts with the grid through the point of connection with the utility. Inside a data center is a distribution grid that connects the utility grid, transformers, UPS units, and loads. The distribution grid is typically a radial network consisting of nodes and their interconnections, where each node represents a specific component such as a UPS, server rack, cooling load, or batteries. Let $n$ be the total number of nodes within the distribution grid. 
The active power injection of each node $i$ is denoted by $p_i$ for $i=1,\cdots,n$, where $p_i>0$ indicates that the node injects power into the network, and $p_i<0$ indicates that the node consumes power. For batteries or UPS equipped with AC/DC inverters, they can also provide reactive power by adjusting the phase angle between AC voltage and current. We denote reactive power of node $i$ as $q_i$, with $q_i=0$ for buses without reactive power injection.

Let $v_0$ be the voltage at the root node, and $v_i, i=1,\cdots,n$ be the voltage of the node $i$ inside data centers. By physical laws of power flow, $v_i$ is jointly determined by the active power $\bm{p}=\left(p_1,\cdots,p_n\right)$ and reactive power $\bm{q}=\left(q_1,\cdots,q_n\right)$ over the  entire network. Assuming the data center operates under a balanced three-phase condition, the voltage dynamics can be approximated using the Linear DistFlow model, given by~\cite{baran_optimal_1989, zhu_fast_2016} 
\begin{align}\label{eq:lindisflow}
    \bm{v} = \bm{R}\bm{p}+\bm{X}\bm{q} + v_0\mathbb{1}
\end{align}
where $\mathbb{1} \in\mathbb{R}^n$ is the vector of ones, and $\bm{R},\bm{X} \in\mathbb{R}^{n\times n}$ are positive definite matrices describing the network topology and parameters (i.e., resistance and reactance). 

We summarize the interconnections between transmission systems and data centers in Figure~\ref{fig:Int_simulation}. Assuming a lossless distribution network, the power injection of data centers to the grid is the summation of power in the internal network $\bm{p}^{DC}=\sum_{i=1}^n p_i$ and $\bm{q}^{DC}=\sum_{i=1}^nq_i$. In turn, voltage $v_0$ of the transmission system impact data center internal voltage through~\eqref{eq:lindisflow}.  In the next section, we establish the control law for $\bm{p}$ and $\bm{q}$ to regulate the data center internal voltage $\bm{v}$ around its nominal value. 
In Section~\ref{sec:experiments}, we will present an integrated test system that simulates both the transient dynamics of the transmission system to faults and the data center distribution network under voltage control. 
% The open-source implementation provides a useful tool for studying power system dynamics with the impact of data centers. 

\begin{figure}[h!]
    \centering
    \includegraphics[width=\linewidth]{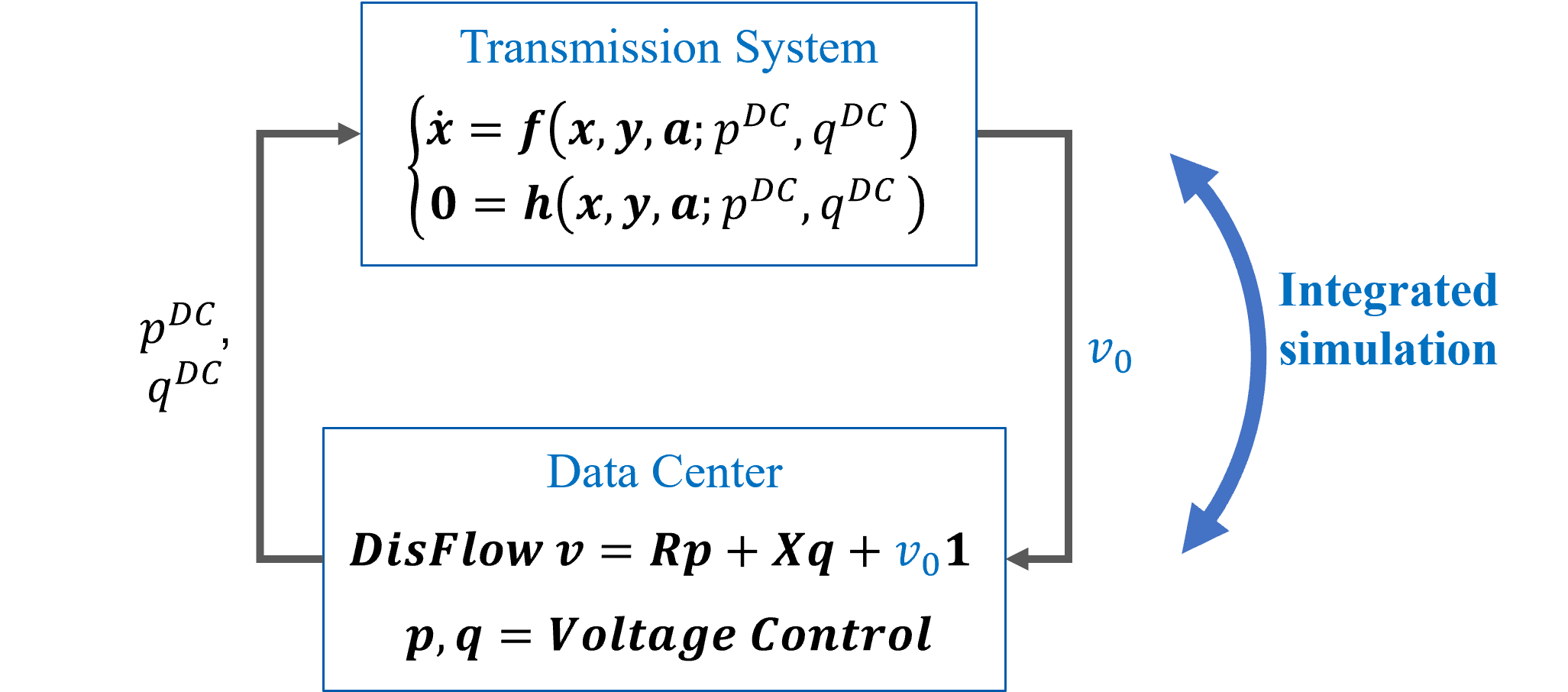}
    \caption{Modeling the interactions between the transmission system and data centers.}    \label{fig:Int_simulation}
\end{figure}
% , and $\bm{p}, \bm{q}\in\mathbb{R}^n$ are the nodal real and reactive power injections (loads and generation).

%%%%%%%%%%%%%%%%%%%%%%%%%%%%%%%%%%%%%
%%%%%%%%%%%%%%%%%%%%%%%%%%%%%%%%%%%%%
\section{Control to Enhance the Low-Voltage Ride-Through Capabilities}
This section illustrates how to enhance the low-voltage ride-through capability of data centers through the design of voltage control laws for their internal distribution systems. We compare centralized and decentralized voltage control strategies so as to allow data center operators to select the suitable approach and construct effective solutions to meet low-voltage ride-through requirements.

\subsection{Centralized Controller}

Centralized controllers coordinate all controllable devices by solving a global optimization problem and simultaneously dispatching the resulting setpoints to each device. This approach relies on low-latency communication to enable real-time data collection and online dispatch across all controllable resources.
% Consider the case where low-latency communication is available across all controllable nodes. 
A centralized optimization problem can be formulated as~\cite{zhang_optimal_2015, yeh_robust_2022, qu_optimal_2020}:
\begin{subequations}\label{eq:centralized}
\begin{align}
\min_{\bm{q}_t, \bm{p}_t} \quad & 
% \sum_{t=0}^T 
\Tilde{\bm{v}}_{t}^\top \bm{Q}_t \Tilde{\bm{v}}_{t} + \bm{q}_t^\top \bm{W}_t^q \bm{q}_t+ \bm{p}_t^\top \bm{W}_t^p \bm{p}_t \\
\text{subject to} \quad 
& \tilde{\bm{v}}_{t} = \bm{R}\bm{p}_t+\bm{X}\bm{q}_t +v_0 \mathbb{1}-\bm{v}^{ref}, \quad t = 0, \ldots, T-1\\
% & \bm{v}_{t+1} = \bm{v}_t + B \bm{q}_t, \quad t = 0, \ldots, T-1 \\
    \quad & \underaccent{\bar}{\bm{q}}_t \leq \bm{q}_t \leq \bar{\bm{q}}_t\\
    \quad & \underaccent{\bar}{\bm{p}}_t \leq \bm{p}_t \leq \bar{\bm{p}}_t,
\end{align}
\end{subequations}
where $\tilde{\bm{v}}_t=\bm{v}_t-\bm{v}^{\text{ref}}$ is the voltage deviation from its reference value. For demonstration purposes, we adopt a quadratic cost function where $\bm{Q}_t , \bm{W}_t^q, \bm{W}_t^p$ denote the weights associated with the costs of voltage deviations, reactive power, and active power, respectively. Any convex cost function can be used without affecting the analytical framework of this paper. The upper and lower bounds for reactive power at the time $t$ is $\underaccent{\bar}{\bm{q}}_t$ and $\bar{\bm{q}}_t$, respectively.
% The bounds are normally determined by the capacities of energy storage and UPS with AC/DC inverters. 
Similarly, the upper and lower bounds for active power at the time $t$ is $\underaccent{\bar}{\bm{p}}_t$ and $\bar{\bm{p}}_t$, respectively. The bounds are determined by the status of controllable devices and the nominal computing load at the time step $t$.

In data centers where low-latency communication network is available for the power delivery infrastructure, and where controllable devices have fast response times, the centralized approach is suitable since it can optimally trade off control cost with disturbance rejection up to device power limits. The effect of delay is studied in Section \ref{sec:experiments}.

%%%%%%%%%%%%%%%%%%%%%%%%%%%%%%%%%%%%%%%%%%%%%%%
\subsection{Decentralized Controller}

The latency in hierarchical control layers and the communication infrastructure in existing data centers may be prohibitively large for centralized control in the timescales of LVRT. To ensure fast and reliable voltage support without relying on centralized coordination, a decentralized control law where each device only uses local information becomes necessary. 
Specifically, one representative decentralized control law is to incrementally adjust the active and reactive power at each node based on its local voltage deviation~\cite{, qu_optimal_2020, cui_decentralized_2022}, written as 
\begin{equation}\label{eq:decentralized_control}
\begin{split}
    p_{i,t} = p_{i,t-1} - k^p_i\tilde{v}_{i,t-1},\\
    q_{i,t} = q_{i,t-1} - k^q_i\tilde{v}_{i,t-1},
\end{split}
\end{equation}
where $k_i^p$ and $k_i^q$ are the tunable control gains for regulating the real and reactive power of each node $i=1,\cdots, n$. 
% the control effort is gradually increased $k_iv_t$ proportional to the deviation in current state $v_t$.
Note that this design requires no real-time communication among nodes. It scales naturally with the size of the data center and remains robust against communication delays or equipment failures.

Despite the decentralized controller design, the voltage at each node is jointly influenced by the actions of all other nodes coupled through the Linear DistFlow model in~\eqref{eq:lindisflow}. By appropriately tuning the decentralized control gains, we can ensure that the collective action of the decentralized controllers drives the voltages toward their reference values. The conditions for voltage convergence are established in the following theorem.

\begin{theorem}[Convergence of voltage]\label{thm:stability}
Let $ \bm{K}^p:=diag(k_1^p,\cdots, k_n^p)$ and $ \bm{K}^q:=diag(k_1^q,\cdots, k_n^q)$ be the diagonal matrices formed by control gains in~\eqref{eq:decentralized_control}.
If the spectral radius of $\left(\bm{I}-\bm{R}\bm{K}^p-\bm{X}\bm{K}^q\right)$ is smaller than 1, 
% (i.e., the dynamical system~\eqref{eq:transition} is stable), 
then the voltage deviation $\tilde{\bm{v}}$ will exponentially converge to zero.
\end{theorem}

\begin{proof}
    Plugging~\eqref{eq:decentralized_control} into~\eqref{eq:lindisflow} yields a dynamical system written as
\begin{equation}\label{eq:transition}
\begin{split}
\tilde{\bm{v}}_{t} &= \bm{R}\left(\bm{p}_{t-1}-\bm{K}^p\tilde{\bm{v}}_{t-1}\right)+\bm{X}\left(\bm{q}_{t-1}-\bm{K}^q\tilde{\bm{v}}_{t-1}\right) +v_0 \mathbb{1}-\bm{v}^{ref} \\
&= \left(\bm{I}-\bm{R}\bm{K}^p-\bm{X}\bm{K}^q\right)\tilde{\bm{v}}_{t}.
\end{split}
\end{equation}
% where $ \bm{K}^p:=diag(k_1^p,\cdots, k_n^p)$ and $ \bm{K}^q:=diag(k_1^q,\cdots, k_n^q)$ are diagonal matrices formed by control gains.
Thus, the incremental control law creates a dynamical system with transition matrix $\left(\bm{I}-\bm{R}\bm{K}^p-\bm{X}\bm{K}^q\right)$ and equilibrium $\tilde{\bm{v}}=0$. The exponential convergence to the equilibrium is guaranteed if the spectral radius of $\left(\bm{I}-\bm{R}\bm{K}^p-\bm{X}\bm{K}^q\right)$ is smaller than 1.
% The voltage deviations are governed by the linear dynamical equation in~\eqref{eq:transition}. The exponential convergence of the voltage deviations to zero is guaranteed if the magnitude of eigenvalues of $\left(\bm{I}-\bm{R}\bm{K}^p-\bm{X}\bm{K}^q\right)$ is smaller than one (i.e., the dynamical system~\eqref{eq:transition} is stable). 
\end{proof}
The condition in Theorem~\ref{thm:stability} can be numerically checked to verify whether the control gains stabilize the voltage. After adding a mild condition on the ratio between resistance and reactance of power lines in data centers, we have the following convex set for stabilizing control gains.  

\begin{theorem}[Decentralized Stabilizing Conditions]\label{thm:decentralized}
Suppose the ratio of resistance to reactance for each power line in the distribution system is $\rho$.  If  $0\prec\rho\bm{K}^p+\bm{K}^q\prec 2\mathbf{X}^{-1}$,
% If $0 \prec \bm{K}^q \prec 2\alpha\mathbf{X}^{-1}$ and $0 \prec \bm{K}^p \prec 2(1-\alpha)\mathbf{R}^{-1}$ for a constant scalar $\alpha\in(0,1)$, 
then the equilibrium point $\tilde{\bm{v}}=0$ of the dynamic system in  \eqref{eq:transition} is locally exponentially stable.
% \wc{It would be wonderful if we can relax the conditon on the fixed r/x ratio. }
% If the magnitude of eigenvalues of $\left(\bm{I}-\bm{R}\bm{K}^p-\bm{X}\bm{K}^q\right)$ is smaller than one (i.e., the dynamical system~\eqref{eq:transition} is stable), then the voltage deviation will exponentially converge to the equilibrium where $\tilde{\bm{v}}=0$.
\end{theorem}
\begin{proof}
By physical law, $\bm{R}=\bm{M}^{-T}\bm{D}_r\bm{M}^{-1}$ and $\bm{X}=\bm{M}^{-T}\bm{D}_x\bm{M}^{-1}$, where $\bm{M}$ is the graph Laplacian matrix of the distribution system, and $\bm{D}_r$ and $\bm{D}_x$ are diagonal matrices formed by stacking the resistance and reactance of power lines~\cite{ zhu_fast_2016}. If the ratio of resistance to reactance for each power line is $\rho$, then $\bm{R}=\rho\bm{X}$ and therefore
    $\bm{I}-\bm{R}\bm{K}^p-\bm{X}\bm{K}^q = \bm{I}-\bm{X}\left(\rho\bm{K}^p+\bm{K}^q\right)$.

    Next, we prove that the eigenvalues of $\bm{I}-\bm{X}\left(\rho\bm{K}^p+\bm{K}^q\right)$ is the same as that of $\bm{I}-\left(\rho\bm{K}^p+\bm{K}^q\right)^{1/2}\bm{X}\left(\rho\bm{K}^p+\bm{K}^q\right)^{1/2}$. Let $\lambda$ be the eigenvalue and $\bm{w}$ be the eigenvector of $\bm{I}-\bm{X}\left(\rho\bm{K}^p+\bm{K}^q\right)$, then $\left(\bm{I}-\bm{X}\left(\rho\bm{K}^p+\bm{K}^q\right)\right)\bm{w}=\lambda\bm{w}$.
Note that $\left(\rho\bm{K}^p+\bm{K}^q\right)$ is a diagonal matrix with positive diagonal elements, we have
\begin{equation*}
\begin{split}
        &\left(\bm{I}-\left(\rho\bm{K}^p+\bm{K}^q\right)^{1/2}\bm{X}\left(\rho\bm{K}^p+\bm{K}^q\right)^{1/2}\right)\left(\rho\bm{K}^p+\bm{K}^q\right)^{1/2}\bm{w}\\
        &=\left(\rho\bm{K}^p+\bm{K}^q\right)^{1/2}\left(\bm{I}-\bm{X}\left(\rho\bm{K}^p+\bm{K}^q\right)\right)\bm{w}\\
        &=\lambda\left(\rho\bm{K}^p+\bm{K}^q\right)^{1/2}\bm{w}.
\end{split}
\end{equation*}
Therefore, $\lambda$ is also the eigenvector of $\bm{I}-\left(\rho\bm{K}^p+\bm{K}^q\right)^{1/2}\bm{X}\left(\rho\bm{K}^p+\bm{K}^q\right)^{1/2}$. To prove that the magnitude of eigenvalues $\lambda$ is smaller than 1, it suffices to show $-\bm{I}\!\prec\!\bm{I}-\left(\rho\bm{K}^p\!+\!\bm{K}^q\right)^{1/2}\!\bm{X}\!\left(\rho\bm{K}^p\!+\!\bm{K}^q\right)^{1/2}\prec\bm{I}$. The right side inequality holds because $\bm{X}\succ0$, while the left-side inequality holds when $0\prec\rho\bm{K}^p+\bm{K}^q\prec 2\mathbf{X}^{-1}$. This concludes the proof. 
\end{proof}

Theorem~\ref{thm:decentralized} characterizes a convex set of stabilizing control gains that can be imposed as constraints when optimizing the controller design. We establish the following optimization program of decentralized control gains to minimize the summation of costs over $T$ steps:
% Optimization for the LTI decentralized controller $K$

\begin{subequations}\label{eq:decentralized_opt}
\begin{align}
\min_{k^{q(p)}_1,\cdots,k^{q(p)}_n} \quad & \sum_{t=0}^T \Tilde{\bm{v}}_{t}^\top \bm{Q}_t \Tilde{\bm{v}}_{t} + \bm{q}_t^\top \bm{W}_t^q \bm{q}_t+ \bm{p}_t^\top \bm{W}_t^p \bm{p}_t  \\
\text{subject to} \quad &     \bm{v}_t = \bm{R}\bm{p}_t+\bm{X}\bm{q}_t + \mathbb{1}, \quad t = 0, \ldots, T-1, \\
& q_{i,t} = q_{i,t-1} - k_i^q(v_{i,t-1}-1),\\
& p_{i,t} = p_{i,t-1} - k_i^p(v_{i,t-1}-1),\\
% & \bm{q}_t=K v_t \\
& k_i^p, k_i^q\text{ is stabilizing},\\
&\underaccent{\bar}{\bm{q}}_t \leq \bm{q}_t \leq \bar{\bm{q}}_t\\
    & \underaccent{\bar}{\bm{p}}_t \leq \bm{p}_t \leq \bar{\bm{p}}_t,\\
\end{align}
\end{subequations}
where the cost function and constraints on actions coincide with that of the centralized optimization problem in~\eqref{eq:centralized}. Note that the optimization is not a standard LQR formulation since the controller is decentralized. Therefore, we adopt the learning-based framework in~\cite{cui_leveraging_2025} for solving~\eqref{eq:decentralized_opt} through gradient descent. 

As long as the control gains satisfies the stability bounds in Theorem~\ref{thm:stability} or Theorem~\ref{thm:decentralized}, the system with local controllers are guaranteed to be exponentially stable. As a corollary, the nodal voltages are exponentially input-to-state stable against disturbances from the transmission system voltage deviations~\cite{cui_leveraging_2025}.

Although the decentralized controller achieves exponential stability, the convergence rate may be slow due to the requirement that a network of local controllers must be stable. This is mitigated if the discrete-time closed-loop controller has a small time step, i.e. 
by updating the control actions more frequently,
% faster updates to the control action, 
so that disturbance rejection is enhanced without compromising stability.

\section{Experimental Results}
\label{sec:experiments}
We evaluate the performance of the proposed control strategies on an integrated test system consisting of a transmission network and a data center power distribution network. We release our code implementation at \href{https://github.com/caltech-netlab/datacenter-voltage-control.git}{https://github.com/caltech-netlab/datacenter-voltage-control.git}, which will provide the research community with a useful tool for studying power system dynamics with data center impacts. 
% To the best of the authors’ knowledge, no open-source code is currently available for conducting joint simulations of power grid transients and data center grids. To fill this gap, we will release our code implementation upon acceptance.
% This will provide the research community with a useful tool for studying power system dynamics with data center impacts. 
%%%%%%%%%%%%%%%%%%%%%%%%%%%%%%%%%%%%%%%%%%%%%%%
\subsection{Simulation Setup}
\begin{table*}[]
    \centering
    \caption{Summary of Data Center Flexible Resources for Case Study}
    \begin{tabular}{c|ccc}
         Load type & Nameplate Capacity (MVA) & Real Power Control Capacity (MW) & Reactive Power Control Capacity (MVar) \\
         \hline
         Utility scale BESS 1 \& 2 & 20 & -20 to 20 & -20 to 20 \\
         Battery UPS 1 to 8 & 5 & -5 to 5 & -5 to 5 \\
         Computing clusters 1 to 8 & 20 & +/-20\% of normal load & 0 \\
         Cooling load & 80 & +/-20\% of normal load & 0 \\
    \end{tabular}
    \label{tab:exp_loads}
\end{table*}

% Interconnection and devices
 The data center distribution network is adapted from the Vulcan Test Platform \cite{vaidhynathan_vulcan_2025}, which is interconnected to the IEEE 14-bus transmission system using the ANDES simulator \cite{cui_hybrid_2021}. 
We consider a network of 1 central cooling plant, 8 data center clusters each with a separate UPS, and 2 utility-scale battery energy storage systems. The nameplate capacity refers to the largest amount of apparent power injection. The amount of controllable power may be less. The values are summarized in Table \ref{tab:exp_loads}. 

% Time-series

Recent years have seen explosive growth in model size and complexity, evolving from early models trained on a single GPU to large language models trained using tens of thousands of GPUs simultaneously \cite{zhu_energy-efficient_2025}. The large power fluctuations due to simultaneous training leads to additional voltage fluctuations, therefore making the VRT problem even more challenging, where the control action must respond to both the external (large) disturbance from the grid and internal (smaller) disturbance from within the data center. The nominal computing load timeseries (Figure \ref{fig:computing_load}) is generated using the GPU utilization profile from large language model inference workloads \cite{patel_characterizing_2024} and multiplied by the nameplate capacity of each computing cluster. The nominal cooling load is assumed to be constant for the time-scale of interest (seconds) and is computed by the average computing utilization (capacity factor) multiplied by the nameplate capacity of the cooling plant.
% \cite{choukse_power_2025}

\begin{figure}[h!]
    \centering
    \includegraphics[width=.9\linewidth]{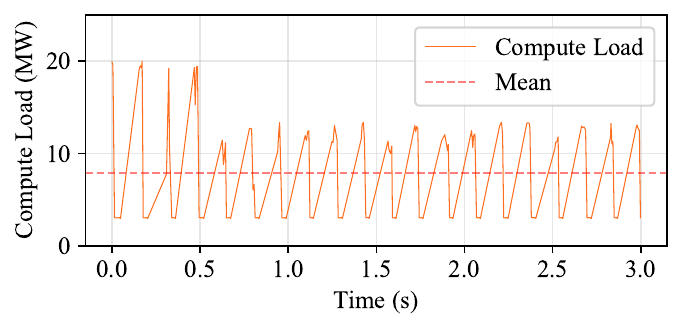}
    \caption{Data cluster computing load (per cluster)}
    \label{fig:computing_load}
\end{figure}

% Other parameters
The transmission system serves a total of 278 MW of load, including 91MW of average load from the data center. Loads aside from the data center are treated as constant-power. The transmission system includes five operating synchronous generators, each with a power rating of 100MVA. We consider a low-voltage event due to the loss of a generator in the high-voltage transmission system. At time $t = 1$ second, a 100 MW generator is disconnected from the transmission system causing a system-wide drop in voltage. 

%%%%%%%%%%%%%%%%%%%%%%%%%%%%%%%%%%%%%%%%%%%%%%%%
\subsection{Baseline: No Voltage Control}
% Time-series
We first establish a baseline scenario in the absence of voltage controllers and with traditional UPS protection mechanism. The data center protection mechanism disconnects all data center loads according to the IEC 62040-3 standard (Figure \ref{fig:VRT_curves}). The voltage trajectories and data center power injection is illustrated in Figure~\ref{fig:exp_high_batt_no_control}. After the loss of generation, the system voltages rapidly drops with no mitigating control actions from the data center. When the trip criteria is met, the data center is disconnected from the grid at $t\approx1.8$ seconds and the system voltages rapidly rise to 1.1 per-unit (p.u.)%
\footnote{Per-unit voltages are the voltage magnitudes normalized by their nominal value. 1.1 per-unit corresponds to 10\% above the nominal value.}
where emergency mitigation actions becomes necessary and further tripping of grid-connected devices may occur.

\begin{figure}[b]
    \centering
    \includegraphics[width=\linewidth]{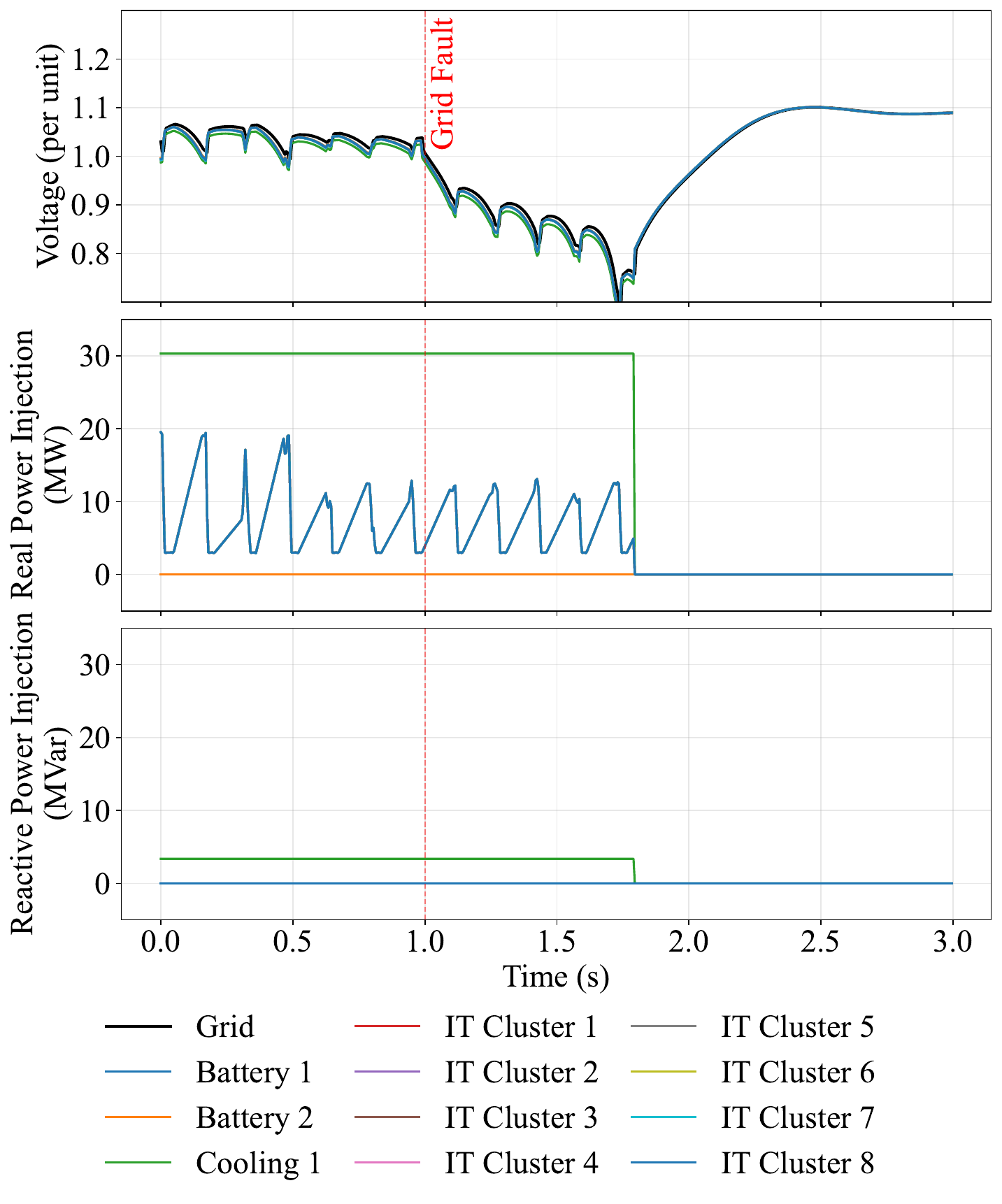}
    \caption{Data center voltages and load without voltage control.}
    \label{fig:exp_high_batt_no_control}
\end{figure}

%%%%%%%%%%%%%%%%%%%%%%%%%%%%%%%%%%%%%%%%%%%%%%%%
\subsection{Centralized and Decentralized Control}
We now apply the proposed centralized and decentralized controllers to stabilize the voltages internal to the data center. For the centralized controller, the control delay is 50 milliseconds, reflecting realistic communication and computation latencies. As shown in Fig.~\ref{fig:exp_high_batt_centralized}, the centralized optimal controller stabilizes bus voltages well within 10\% of the nominal voltage. Despite the delay, the controller is able to damp oscillations and prevent tripping of UPS, thus providing fault ride-through capability.

In comparison, under the decentralized control strategy, each distributed energy resource adjusts its response based solely on local measurements. The decentralized control gains satisfy the stabilizing conditions in Section~\ref{thm:decentralized}, and we optimize the set of stable control gains via gradient descent through the learning-based framework in~\cite{cui_leveraging_2025}. For decentralized controllers, the control actions are updated at discrete time intervals of 5 milliseconds. The experimental results are shown in Fig.~\ref{fig:exp_high_batt_decentralized}. Despite the lack of centralized coordination, the decentralized controller achieves effective voltage stabilization within the data center distribution system. However, the decentralized scheme requires participation from more types of controllable resources due to the lack of centralized coordination.

\begin{table*}[]
    \centering
    \caption{Performance Evaluation of Different Control Schemes}
    \begin{tabular}{p{4.0cm}|p{2.4cm}p{2.2cm}p{3.1cm}p{3.4cm}}
         \textbf{Control Scheme} & \textbf{Largest Voltage Deviation (p.u.)} & \textbf{Mean Voltage Deviation (p.u.)} & \textbf{Average Real Power Control Effort (MW)} & \textbf{Average Reactive Power Control Effort (MVAr)} \\
         \hline
         No Voltage Control & 0.383 & 0.081 & 0.030 & 0.001 \\
         Centralized & 0.061 & 0.019 & 0.012 & 0.006 \\
         Decentralized & 0.087 & 0.016 & 0.007 & 0.006 \\
         Centralized (reactive power) & 0.075 & 0.025 & 0.000 & 0.008 \\
         Decentralized (reactive power) & 0.125 & 0.034 & 0.000 & 0.002 \\
         Centralized (200ms delay) & 0.117 & 0.039 & 0.010 & 0.010 \\
    \end{tabular}
    \label{tab:exp_metrics}
\end{table*}

\begin{figure}[b]
    \centering
    \includegraphics[width=1.0\linewidth]{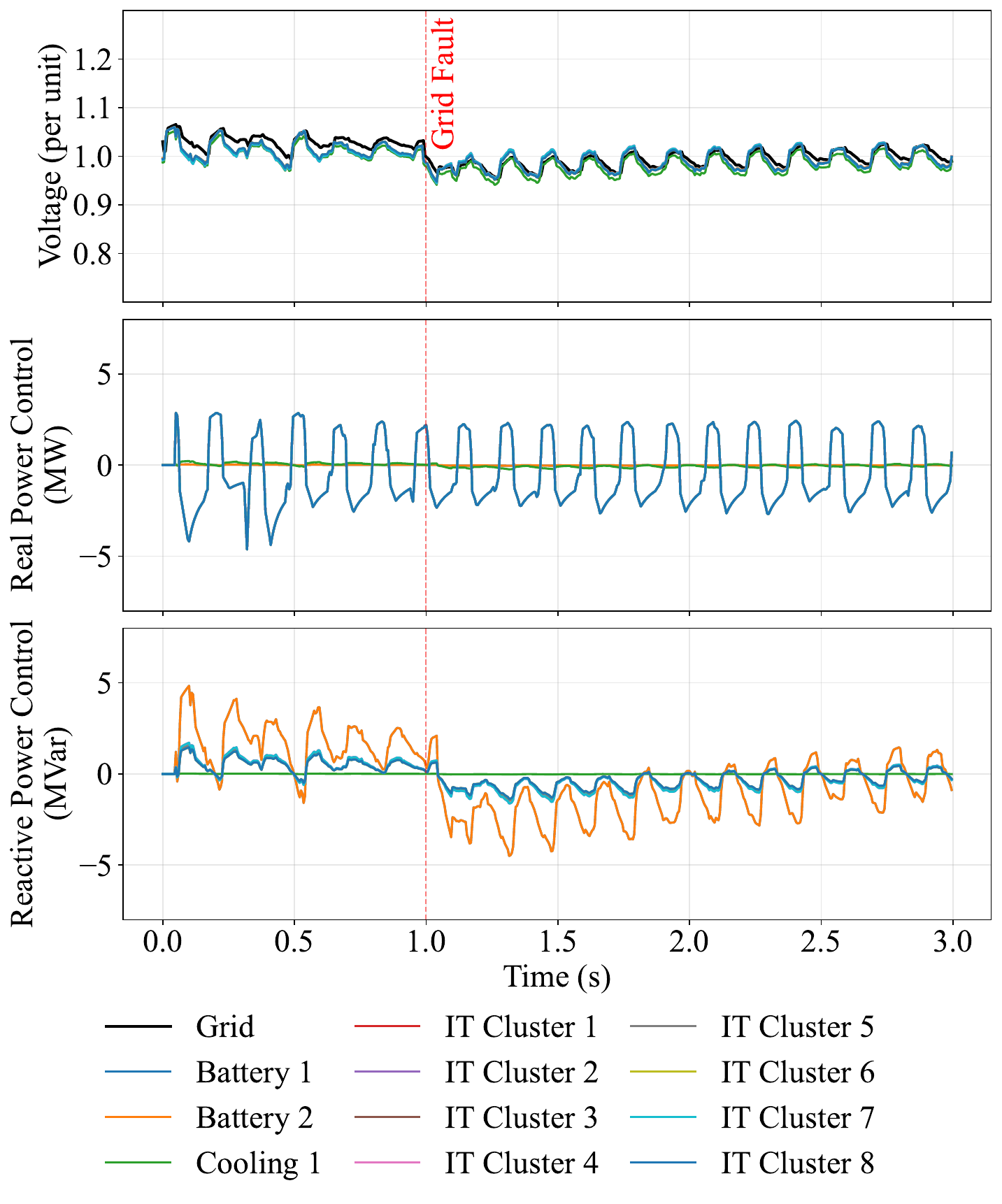}
    \caption{Data center voltages and control action with centralized controller.}
    \label{fig:exp_high_batt_centralized}
\end{figure}

\begin{figure}[b]
    \centering
    \includegraphics[width=1.0\linewidth]{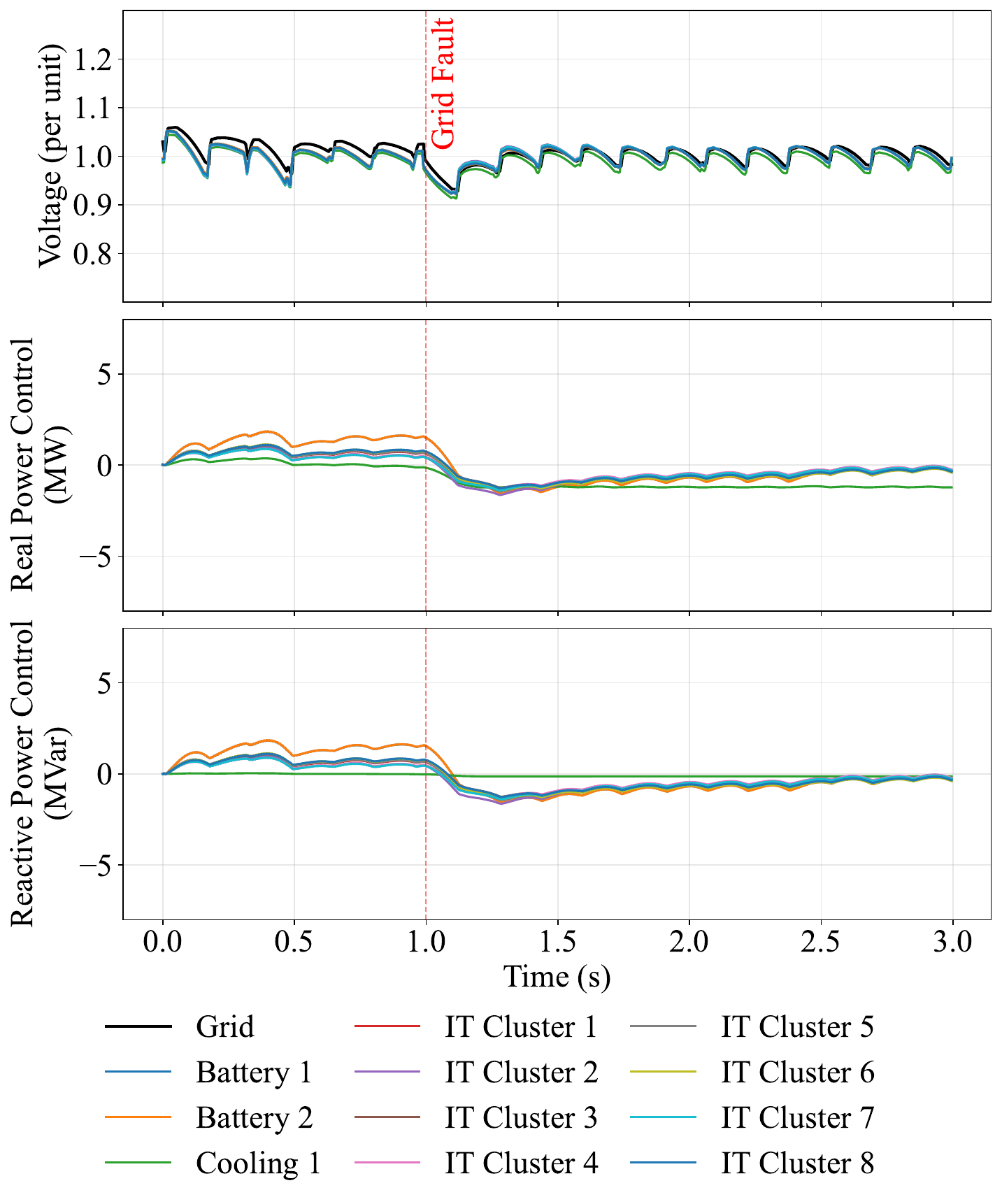}
    \caption{Data center voltages and control action with decentralized controller.}
    \label{fig:exp_high_batt_decentralized}
\end{figure}

%%%%%%%%%%%%%%%%%%%%%%%%%%%%%%%%%%%%%%%%%%%%%%%%
\subsection{Effect of Delay on Centralized Control}
Next, we study the impact of delay in the centralized controller. Figure \ref{fig:delay_ablation} shows that the mean and maximum voltage deviation increases with delay, and the control effort generally increases with delay. The performance degrades gracefully. In this setting, compared to the de-centralized controller (Table \ref{tab:exp_metrics}), the centralized control is worse for delays longer than 10 milliseconds. In reality, the performance of the decentralized controller also depends on the discrete time interval at which the control actions are updated. The choice of the suitable controller depends on the real-world hardware limitations.

\begin{figure}[b]
    \centering
    \includegraphics[width=1.0\linewidth]{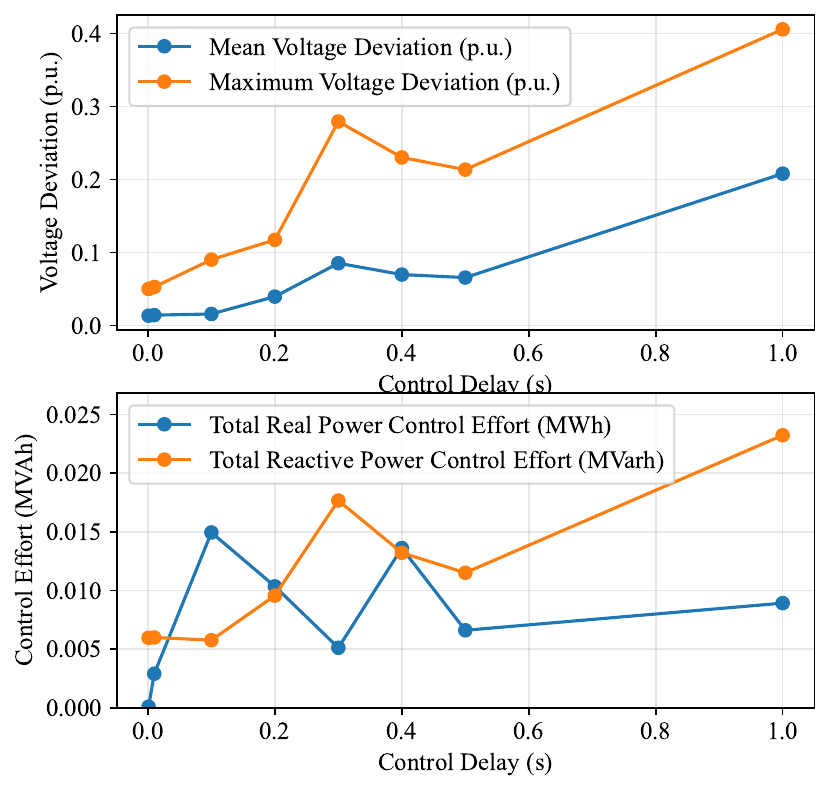}
    \caption{Impact of delay on performance metrics for centralized controller.}
    \label{fig:delay_ablation}
\end{figure}

%%%%%%%%%%%%%%%%%%%%%%%%%%%%%%%%%%%%%%%%%%%%%%%
\subsection{Volt-Var Control}
Finally, we consider the case where real power cannot be controlled. This is relevant when the battery state of charge is important and cannot be changed, and the computing and cooling loads cannot be controlled. The only controllable injections are the reactive power from BESS and UPS (i.e. volt-var control). Table \ref{tab:exp_metrics} shows that for both the centralized and decentralized controllers, the performance degrades when the controllability in real power is removed, although the voltages are still stabilized within the acceptable envelope of IEC 62040-3 curve (Figure \ref{fig:VRT_curves}), so that the data center remains connected to the grid.

\section{Conclusion}
In this work, we propose enhancing the low voltage ride-through capability of data centers by implementing voltage control within their internal distribution networks. We examine the low-voltage ride-through standards relevant to data centers and systematically analyze the controllable resources available in these facilities. Building on this foundation, we model the distribution networks of data centers and compare both centralized and decentralized control solutions.
The experiments confirm that uncontrolled operation leads to tripping due to voltage violations, while both centralized and decentralized controllers can provide low-voltage ride through in the integrated IEEE 14-bus and the data center systems. The centralized approach provides globally optimal control action when delay is negligible, but the performance degrades with actuation delay. It is therefore best suited to settings with low-latency communication infrastructure, which is typically available within a data center.
% Therefore it is suitable in the setting where low-latency communication infrastructure is available, which is typically the case within a data center.
In contrast, the decentralized method offers a scalable, communication-free alternative. However, the iterative control law requires more steps to converge, since the incremental control gain needs to be sufficiently conservative in order to maintain network stability.

Future work includes evaluating the costs and trade-offs associated with data center services (e.g., power quality and potential impacts on data center operations) and assessing how these costs influence FRT capabilities. In addition, incorporating the dynamic responses of controllable devices and exploring further interactions with other grid services (e.g., under-frequency load shedding and ramp-rate limiting) are also interesting directions for continued research.

% In addition, we will study 
% systematically integrating grid services including post-fault active Power Recovery, under frequency load shedding, etc..
% Future work: systematically combine grid services (Post Fault Active Power Recovery, FFR, FRT, DR, under frequency load shedding, ramp rate constraint) and data center services (backup, power quality)

% \begin{acks}
% \end{acks}

\bibliographystyle{ACM-Reference-Format}
\bibliography{main}


\begin{thebibliography}{66}

%%% the Web address.

\ifx \showCODEN    \undefined \def \showCODEN     #1{\unskip}     \fi
\ifx \showISBNx    \undefined \def \showISBNx     #1{\unskip}     \fi
\ifx \showISBNxiii \undefined \def \showISBNxiii  #1{\unskip}     \fi
\ifx \showISSN     \undefined \def \showISSN      #1{\unskip}     \fi
\ifx \showLCCN     \undefined \def \showLCCN      #1{\unskip}     \fi
\ifx \shownote     \undefined \def \shownote      #1{#1}          \fi
\ifx \showarticletitle \undefined \def \showarticletitle #1{#1}   \fi
\ifx \showURL      \undefined \def \showURL       {\relax}        \fi
% invisible to TeX
\providecommand\bibfield[2]{#2}
\providecommand\bibinfo[2]{#2}
\providecommand\natexlab[1]{#1}
\providecommand\showeprint[2][]{arXiv:#2}

\bibitem[noa({[n.\,d.]})]%
        {noauthor_iti_nodate}
 \bibinfo{year}{[n.\,d.]}\natexlab{}.
\newblock \bibinfo{title}{{ITI} ({CBEMA}) {Curve} {Application} {Note}}.
\newblock
\urldef\tempurl%
\url{https://users.etown.edu/w/wunderjt/PACKET%2034%20CLEAN%20POWER.PDF}
\showURL{%
\tempurl}


\bibitem[iec(2011)]%
        {iec_uninterruptible_2011}
 \bibinfo{year}{2011}\natexlab{}.
\newblock \showarticletitle{Uninterruptible power systems ({UPS}):
  international standard = {Alimentations} sans interruption ({ASI}). {Pt}. 3:
  {Method} of specifying the performance and test requirements = {Méthode} de
  spécification des performances et exigences d'essais}.
\newblock Number 62040-3 in \bibinfo{series}{{IEC} international standard /
  {International} {Electrotechnical} {Commission}}. \bibinfo{publisher}{IEC,
  Int. Electrotechnical Commission}, \bibinfo{address}{Geneva}.
\newblock
\showISBNx{978-2-88912-384-1}
\newblock
\shownote{Num Pages: 214}.


\bibitem[Ahmed et~al\mbox{.}(2021)]%
        {ahmed_review_2021}
\bibfield{author}{\bibinfo{person}{Kazi Main~Uddin Ahmed},
  \bibinfo{person}{Math H.~J. Bollen}, {and} \bibinfo{person}{Manuel Alvarez}.}
  \bibinfo{year}{2021}\natexlab{}.
\newblock \showarticletitle{A {Review} of {Data} {Centers} {Energy}
  {Consumption} and {Reliability} {Modeling}}.
\newblock \bibinfo{journal}{\emph{IEEE Access}}  \bibinfo{volume}{9}
  (\bibinfo{year}{2021}), \bibinfo{pages}{152536--152563}.
\newblock
\showISSN{2169-3536}
\href{https://doi.org/10.1109/ACCESS.2021.3125092}{doi:\nolinkurl{10.1109/ACCESS.2021.3125092}}


\bibitem[{Babu Chalamala} et~al\mbox{.}(2025)]%
        {babu_chalamala_data_2025}
\bibfield{author}{\bibinfo{person}{{Babu Chalamala}},
  \bibinfo{person}{{Veronika Rabl}}, \bibinfo{person}{{Sami Abdulsalam}},
  \bibinfo{person}{{Oluwatimilehin (Timi) Adeosun}}, \bibinfo{person}{{Scott
  Baker}}, \bibinfo{person}{{Lina Bertling Tjernberg}},
  \bibinfo{person}{{Valerie Carter-Ridley}}, \bibinfo{person}{{Tim Horger}},
  \bibinfo{person}{{Kamal Garg}}, \bibinfo{person}{{Andrew Gledhill}},
  \bibinfo{person}{{Bill Hederman}}, \bibinfo{person}{{Richard Kirby}},
  \bibinfo{person}{{Collin Martin}}, \bibinfo{person}{{Molly Mooney}},
  \bibinfo{person}{{Shaun Moran}}, \bibinfo{person}{{Manish Patel}},
  \bibinfo{person}{{Tom Pierpoint}}, \bibinfo{person}{{Eddie Schweitzer}},
  \bibinfo{person}{{Kevin Sherd}}, \bibinfo{person}{{JP Skeath}}, {and}
  \bibinfo{person}{{Luka Strezoski}}.} \bibinfo{year}{2025}\natexlab{}.
\newblock \showarticletitle{Data {Center} {Growth} and {Grid} {Readiness}
  ({TR131})}.
\newblock  (\bibinfo{date}{May} \bibinfo{year}{2025}).
\newblock
\href{https://doi.org/10.17023/W4WY-S557}{doi:\nolinkurl{10.17023/W4WY-S557}}
\newblock
\shownote{Publisher: IEEE}.


\bibitem[Baker et~al\mbox{.}(2017)]%
        {baker_network-cognizant_2017}
\bibfield{author}{\bibinfo{person}{Kyri Baker}, \bibinfo{person}{Andrey
  Bernstein}, \bibinfo{person}{Emiliano Dall'Anese}, {and}
  \bibinfo{person}{Changhong Zhao}.} \bibinfo{year}{2017}\natexlab{}.
\newblock \bibinfo{title}{Network-{Cognizant} {Voltage} {Droop} {Control} for
  {Distribution} {Grids}}.
\newblock
\href{https://doi.org/10.48550/arXiv.1702.02969}{doi:\nolinkurl{10.48550/arXiv.1702.02969}}
\newblock
\shownote{arXiv:1702.02969 [math]}.


\bibitem[Baran and Wu(1989)]%
        {baran_optimal_1989}
\bibfield{author}{\bibinfo{person}{M.E. Baran} {and} \bibinfo{person}{F.F.
  Wu}.} \bibinfo{year}{1989}\natexlab{}.
\newblock \showarticletitle{Optimal capacitor placement on radial distribution
  systems}.
\newblock \bibinfo{journal}{\emph{IEEE Transactions on Power Delivery}}
  \bibinfo{volume}{4}, \bibinfo{number}{1} (\bibinfo{date}{Jan.}
  \bibinfo{year}{1989}), \bibinfo{pages}{725--734}.
\newblock
\showISSN{1937-4208}
\href{https://doi.org/10.1109/61.19265}{doi:\nolinkurl{10.1109/61.19265}}
\newblock
\shownote{Conference Name: IEEE Transactions on Power Delivery}.


\bibitem[Bian et~al\mbox{.}(2024)]%
        {bian_cafe_2024}
\bibfield{author}{\bibinfo{person}{Jieming Bian}, \bibinfo{person}{Lei Wang},
  \bibinfo{person}{Shaolei Ren}, {and} \bibinfo{person}{Jie Xu}.}
  \bibinfo{year}{2024}\natexlab{}.
\newblock \showarticletitle{{CAFE}: {Carbon}-{Aware} {Federated} {Learning} in
  {Geographically} {Distributed} {Data} {Centers}✱}. In
  \bibinfo{booktitle}{\emph{The 15th {ACM} {International} {Conference} on
  {Future} and {Sustainable} {Energy} {Systems}}}. \bibinfo{publisher}{ACM},
  \bibinfo{address}{Singapore Singapore}, \bibinfo{pages}{347--360}.
\newblock
\showISBNx{979-8-4007-0480-2}
\href{https://doi.org/10.1145/3632775.3661970}{doi:\nolinkurl{10.1145/3632775.3661970}}


\bibitem[Bolognani and Zampieri(2013)]%
        {bolognani_distributed_2013}
\bibfield{author}{\bibinfo{person}{Saverio Bolognani} {and}
  \bibinfo{person}{Sandro Zampieri}.} \bibinfo{year}{2013}\natexlab{}.
\newblock \showarticletitle{A {Distributed} {Control} {Strategy} for {Reactive}
  {Power} {Compensation} in {Smart} {Microgrids}}.
\newblock \bibinfo{journal}{\emph{IEEE Trans. Automat. Control}}
  \bibinfo{volume}{58}, \bibinfo{number}{11} (\bibinfo{date}{Nov.}
  \bibinfo{year}{2013}), \bibinfo{pages}{2818--2833}.
\newblock
\showISSN{1558-2523}
\href{https://doi.org/10.1109/TAC.2013.2270317}{doi:\nolinkurl{10.1109/TAC.2013.2270317}}


\bibitem[Bostandoost et~al\mbox{.}(2024)]%
        {bostandoost_lacs_2024}
\bibfield{author}{\bibinfo{person}{Roozbeh Bostandoost}, \bibinfo{person}{Adam
  Lechowicz}, \bibinfo{person}{Walid~A. Hanafy}, \bibinfo{person}{Noman
  Bashir}, \bibinfo{person}{Prashant Shenoy}, {and} \bibinfo{person}{Mohammad
  Hajiesmaili}.} \bibinfo{year}{2024}\natexlab{}.
\newblock \showarticletitle{{LACS}: {Learning}-{Augmented} {Algorithms} for
  {Carbon}-{Aware} {Resource} {Scaling} with {Uncertain} {Demand}}. In
  \bibinfo{booktitle}{\emph{The 15th {ACM} {International} {Conference} on
  {Future} and {Sustainable} {Energy} {Systems}}}. \bibinfo{publisher}{ACM},
  \bibinfo{address}{Singapore Singapore}, \bibinfo{pages}{27--45}.
\newblock
\showISBNx{979-8-4007-0480-2}
\href{https://doi.org/10.1145/3632775.3661942}{doi:\nolinkurl{10.1145/3632775.3661942}}


\bibitem[Cao et~al\mbox{.}(2023)]%
        {cao_toward_2023}
\bibfield{author}{\bibinfo{person}{Zhiwei Cao}, \bibinfo{person}{Ruihang Wang},
  \bibinfo{person}{Xin Zhou}, {and} \bibinfo{person}{Yonggang Wen}.}
  \bibinfo{year}{2023}\natexlab{}.
\newblock \showarticletitle{Toward {Model}-{Assisted} {Safe} {Reinforcement}
  {Learning} for {Data} {Center} {Cooling} {Control}: {A} {Lyapunov}-based
  {Approach}}. In \bibinfo{booktitle}{\emph{Proceedings of the 14th {ACM}
  {International} {Conference} on {Future} {Energy} {Systems}}}.
  \bibinfo{publisher}{ACM}, \bibinfo{address}{Orlando FL USA},
  \bibinfo{pages}{333--346}.
\newblock
\showISBNx{979-8-4007-0032-3}
\href{https://doi.org/10.1145/3575813.3597343}{doi:\nolinkurl{10.1145/3575813.3597343}}


\bibitem[Chen et~al\mbox{.}(2020)]%
        {chen_data-driven_2020}
\bibfield{author}{\bibinfo{person}{Yize Chen}, \bibinfo{person}{Yuanyuan Shi},
  {and} \bibinfo{person}{Baosen Zhang}.} \bibinfo{year}{2020}\natexlab{}.
\newblock \showarticletitle{Data-{Driven} {Optimal} {Voltage} {Regulation}
  {Using} {Input} {Convex} {Neural} {Networks}}.
\newblock \bibinfo{journal}{\emph{Electric Power Systems Research}}
  \bibinfo{volume}{189} (\bibinfo{date}{Dec.} \bibinfo{year}{2020}),
  \bibinfo{pages}{106741}.
\newblock
\showISSN{0378-7796}
\href{https://doi.org/10.1016/j.epsr.2020.106741}{doi:\nolinkurl{10.1016/j.epsr.2020.106741}}


\bibitem[Chen and Zhang(2025)]%
        {chen_voltage_2025}
\bibfield{author}{\bibinfo{person}{Yize Chen} {and} \bibinfo{person}{Baosen
  Zhang}.} \bibinfo{year}{2025}\natexlab{}.
\newblock \bibinfo{title}{Voltage {Regulation} in {Distribution} {Systems} with
  {Data} {Center} {Loads}}.
\newblock
\href{https://doi.org/10.48550/arXiv.2507.06416}{doi:\nolinkurl{10.48550/arXiv.2507.06416}}
\newblock
\shownote{arXiv:2507.06416 [eess]}.


\bibitem[Chiang({[n.\,d.]})]%
        {chiang_direct_nodate}
\bibfield{author}{\bibinfo{person}{Hsiao-Dong Chiang}.}
  \bibinfo{year}{[n.\,d.]}\natexlab{}.
\newblock \bibinfo{booktitle}{\emph{Direct {Methods} for {Stability} {Analysis}
  of {Electric} {Power} {Systems}: {Theoretical} {Foundation}, {BCU}
  {Methodologies}, and {Applications} {\textbar} {Wiley}}}.
\newblock
\urldef\tempurl%
\url{https://www.wiley.com/en-us/Direct+Methods+for+Stability+Analysis+of+Electric+Power+Systems%3A+Theoretical+Foundation%2C+BCU+Methodologies%2C+and+Applications-p-9780470484401}
\showURL{%
\tempurl}


\bibitem[Choukse et~al\mbox{.}(2025)]%
        {choukse_power_2025}
\bibfield{author}{\bibinfo{person}{Esha Choukse}, \bibinfo{person}{Brijesh
  Warrier}, \bibinfo{person}{Scot Heath}, \bibinfo{person}{Luz Belmont},
  \bibinfo{person}{April Zhao}, \bibinfo{person}{Hassan~Ali Khan},
  \bibinfo{person}{Brian Harry}, \bibinfo{person}{Matthew Kappel},
  \bibinfo{person}{Russell~J. Hewett}, \bibinfo{person}{Kushal Datta},
  \bibinfo{person}{Yu Pei}, \bibinfo{person}{Caroline Lichtenberger},
  \bibinfo{person}{John Siegler}, \bibinfo{person}{David Lukofsky},
  \bibinfo{person}{Zaid Kahn}, \bibinfo{person}{Gurpreet Sahota},
  \bibinfo{person}{Andy Sullivan}, \bibinfo{person}{Charles Frederick},
  \bibinfo{person}{Hien Thai}, \bibinfo{person}{Rebecca Naughton},
  \bibinfo{person}{Daniel Jurnove}, \bibinfo{person}{Justin Harp},
  \bibinfo{person}{Reid Carper}, \bibinfo{person}{Nithish Mahalingam},
  \bibinfo{person}{Srini Varkala}, \bibinfo{person}{Alok~Gautam Kumbhare},
  \bibinfo{person}{Satyajit Desai}, \bibinfo{person}{Venkatesh Ramamurthy},
  \bibinfo{person}{Praneeth Gottumukkala}, \bibinfo{person}{Girish Bhatia},
  \bibinfo{person}{Kelsey Wildstone}, \bibinfo{person}{Laurentiu Olariu},
  \bibinfo{person}{Ileana Incorvaia}, \bibinfo{person}{Alex Wetmore},
  \bibinfo{person}{Prabhat Ram}, \bibinfo{person}{Melur Raghuraman},
  \bibinfo{person}{Mohammed Ayna}, \bibinfo{person}{Mike Kendrick},
  \bibinfo{person}{Ricardo Bianchini}, \bibinfo{person}{Aaron Hurst},
  \bibinfo{person}{Reza Zamani}, \bibinfo{person}{Xin Li},
  \bibinfo{person}{Michael Petrov}, \bibinfo{person}{Gene Oden},
  \bibinfo{person}{Rory Carmichael}, \bibinfo{person}{Tom Li},
  \bibinfo{person}{Apoorv Gupta}, \bibinfo{person}{Pratikkumar Patel},
  \bibinfo{person}{Nilesh Dattani}, \bibinfo{person}{Lawrence Marwong},
  \bibinfo{person}{Rob Nertney}, \bibinfo{person}{Hirofumi Kobayashi},
  \bibinfo{person}{Jeff Liott}, \bibinfo{person}{Miro Enev},
  \bibinfo{person}{Divya Ramakrishnan}, \bibinfo{person}{Ian Buck}, {and}
  \bibinfo{person}{Jonah Alben}.} \bibinfo{year}{2025}\natexlab{}.
\newblock \bibinfo{title}{Power {Stabilization} for {AI} {Training}
  {Datacenters}}.
\newblock
\href{https://doi.org/10.48550/arXiv.2508.14318}{doi:\nolinkurl{10.48550/arXiv.2508.14318}}
\newblock
\shownote{arXiv:2508.14318 [cs]}.


\bibitem[Cui et~al\mbox{.}(2021)]%
        {cui_hybrid_2021}
\bibfield{author}{\bibinfo{person}{Hantao Cui}, \bibinfo{person}{Fangxing Li},
  {and} \bibinfo{person}{Kevin Tomsovic}.} \bibinfo{year}{2021}\natexlab{}.
\newblock \showarticletitle{Hybrid {Symbolic}-{Numeric} {Framework} for {Power}
  {System} {Modeling} and {Analysis}}.
\newblock \bibinfo{journal}{\emph{IEEE Transactions on Power Systems}}
  \bibinfo{volume}{36}, \bibinfo{number}{2} (\bibinfo{date}{March}
  \bibinfo{year}{2021}), \bibinfo{pages}{1373--1384}.
\newblock
\showISSN{1558-0679}
\href{https://doi.org/10.1109/TPWRS.2020.3017019}{doi:\nolinkurl{10.1109/TPWRS.2020.3017019}}


\bibitem[Cui et~al\mbox{.}(2022)]%
        {cui_decentralized_2022}
\bibfield{author}{\bibinfo{person}{Wenqi Cui}, \bibinfo{person}{Jiayi Li},
  {and} \bibinfo{person}{Baosen Zhang}.} \bibinfo{year}{2022}\natexlab{}.
\newblock \showarticletitle{Decentralized safe reinforcement learning for
  inverter-based voltage control}.
\newblock \bibinfo{journal}{\emph{Electric Power Systems Research}}
  \bibinfo{volume}{211} (\bibinfo{date}{Oct.} \bibinfo{year}{2022}),
  \bibinfo{pages}{108609}.
\newblock
\showISSN{03787796}
\href{https://doi.org/10.1016/j.epsr.2022.108609}{doi:\nolinkurl{10.1016/j.epsr.2022.108609}}


\bibitem[Cui et~al\mbox{.}(2025)]%
        {cui_leveraging_2025}
\bibfield{author}{\bibinfo{person}{Wenqi Cui}, \bibinfo{person}{Yiheng Xie},
  \bibinfo{person}{Steven Low}, \bibinfo{person}{Adam Wierman}, {and}
  \bibinfo{person}{Baosen Zhang}.} \bibinfo{year}{2025}\natexlab{}.
\newblock \bibinfo{title}{Leveraging {Predictions} in {Power} {System}
  {Voltage} {Control}: {An} {Adaptive} {Approach}}.
\newblock
\href{https://doi.org/10.48550/arXiv.2509.09937}{doi:\nolinkurl{10.48550/arXiv.2509.09937}}
\newblock
\shownote{arXiv:2509.09937 [eess]}.


\bibitem[Cui et~al\mbox{.}(2024)]%
        {cui_frequency_2024}
\bibfield{author}{\bibinfo{person}{Wenqi Cui}, \bibinfo{person}{Weiwei Yang},
  {and} \bibinfo{person}{Baosen Zhang}.} \bibinfo{year}{2024}\natexlab{}.
\newblock \showarticletitle{A {Frequency} {Domain} {Approach} to {Predict}
  {Power} {System} {Transients}}.
\newblock \bibinfo{journal}{\emph{IEEE Transactions on Power Systems}}
  \bibinfo{volume}{39}, \bibinfo{number}{1} (\bibinfo{date}{Jan.}
  \bibinfo{year}{2024}), \bibinfo{pages}{465--477}.
\newblock
\showISSN{1558-0679}
\href{https://doi.org/10.1109/TPWRS.2023.3259960}{doi:\nolinkurl{10.1109/TPWRS.2023.3259960}}


\bibitem[Dayarathna et~al\mbox{.}(2016)]%
        {dayarathna_data_2016}
\bibfield{author}{\bibinfo{person}{Miyuru Dayarathna},
  \bibinfo{person}{Yonggang Wen}, {and} \bibinfo{person}{Rui Fan}.}
  \bibinfo{year}{2016}\natexlab{}.
\newblock \showarticletitle{Data {Center} {Energy} {Consumption} {Modeling}:
  {A} {Survey}}.
\newblock \bibinfo{journal}{\emph{IEEE Communications Surveys \& Tutorials}}
  \bibinfo{volume}{18}, \bibinfo{number}{1} (\bibinfo{year}{2016}),
  \bibinfo{pages}{732--794}.
\newblock
\showISSN{1553-877X}
\href{https://doi.org/10.1109/COMST.2015.2481183}{doi:\nolinkurl{10.1109/COMST.2015.2481183}}


\bibitem[{EirGrid}(2025)]%
        {eirgrid_voltage_2025}
\bibfield{author}{\bibinfo{person}{{EirGrid}}.}
  \bibinfo{year}{2025}\natexlab{}.
\newblock \bibinfo{title}{Voltage {Fault} {Ride}-{Through} ({VFRT}) {Study}
  {Template} and {Assessment} {Guide}}.
\newblock
\urldef\tempurl%
\url{https://cms.eirgrid.ie/sites/default/files/publications/Voltage-Fault-Ride-Through-VFRT-Study-Template-and-Assessment-Guide-v3.pdf}
\showURL{%
\tempurl}


\bibitem[{EPRI}({[n.\,d.]})]%
        {epri_dcflex_nodate}
\bibfield{author}{\bibinfo{person}{{EPRI}}.}
  \bibinfo{year}{[n.\,d.]}\natexlab{}.
\newblock \bibinfo{title}{{DCFlex}}.
\newblock
\urldef\tempurl%
\url{https://dcflex.sf.epri.com/}
\showURL{%
\tempurl}


\bibitem[Farivar et~al\mbox{.}(2013)]%
        {farivar_equilibrium_2013}
\bibfield{author}{\bibinfo{person}{Masoud Farivar}, \bibinfo{person}{{Lijun
  Chen}}, {and} \bibinfo{person}{Steven Low}.} \bibinfo{year}{2013}\natexlab{}.
\newblock \showarticletitle{Equilibrium and dynamics of local voltage control
  in distribution systems}. In \bibinfo{booktitle}{\emph{52nd {IEEE}
  {Conference} on {Decision} and {Control}}}. \bibinfo{publisher}{IEEE},
  \bibinfo{address}{Firenze}.
\newblock
\href{https://doi.org/10.1109/cdc.2013.6760555}{doi:\nolinkurl{10.1109/cdc.2013.6760555}}


\bibitem[Feng et~al\mbox{.}(2023)]%
        {feng_bridging_2023}
\bibfield{author}{\bibinfo{person}{Jie Feng}, \bibinfo{person}{Wenqi Cui},
  \bibinfo{person}{Jorge Cortés}, {and} \bibinfo{person}{Yuanyuan Shi}.}
  \bibinfo{year}{2023}\natexlab{}.
\newblock \showarticletitle{Bridging {Transient} and {Steady}-{State}
  {Performance} in {Voltage} {Control}: {A} {Reinforcement} {Learning}
  {Approach} {With} {Safe} {Gradient} {Flow}}.
\newblock \bibinfo{journal}{\emph{IEEE Control Systems Letters}}
  \bibinfo{volume}{7} (\bibinfo{year}{2023}), \bibinfo{pages}{2845--2850}.
\newblock
\showISSN{2475-1456}
\href{https://doi.org/10.1109/LCSYS.2023.3289435}{doi:\nolinkurl{10.1109/LCSYS.2023.3289435}}


\bibitem[Filippi and Valentini(2023)]%
        {filippi_how_2023}
\bibfield{author}{\bibinfo{person}{Arturo~Di Filippi} {and}
  \bibinfo{person}{Luca Valentini}.} \bibinfo{year}{2023}\natexlab{}.
\newblock \showarticletitle{How to {Maximize} {Revenues} from {Your} {Data}
  {Center} {Energy} {Storage} {System} with {Grid} {Interactive} {UPS}}.
\newblock  (\bibinfo{year}{2023}).
\newblock


\bibitem[Fu et~al\mbox{.}(2020)]%
        {fu_assessments_2020}
\bibfield{author}{\bibinfo{person}{Yangyang Fu}, \bibinfo{person}{Xu Han},
  \bibinfo{person}{Kyri Baker}, {and} \bibinfo{person}{Wangda Zuo}.}
  \bibinfo{year}{2020}\natexlab{}.
\newblock \showarticletitle{Assessments of data centers for provision of
  frequency regulation}.
\newblock \bibinfo{journal}{\emph{Applied Energy}}  \bibinfo{volume}{277}
  (\bibinfo{date}{Nov.} \bibinfo{year}{2020}), \bibinfo{pages}{115621}.
\newblock
\showISSN{03062619}
\href{https://doi.org/10.1016/j.apenergy.2020.115621}{doi:\nolinkurl{10.1016/j.apenergy.2020.115621}}


\bibitem[Ghatikar et~al\mbox{.}(2012)]%
        {ghatikar_demand_2012}
\bibfield{author}{\bibinfo{person}{Girish Ghatikar}, \bibinfo{person}{Venkata
  Ganti}, \bibinfo{person}{Nance Matson}, {and} \bibinfo{person}{Mary~Ann
  Piette}.} \bibinfo{year}{2012}\natexlab{}.
\newblock \bibinfo{booktitle}{\emph{Demand {Response} {Opportunities} and
  {Enabling} {Technologies} for {Data} {Centers}: {Findings} {From} {Field}
  {Studies}}}.
\newblock \bibinfo{type}{{T}echnical {R}eport} LBNL--5763E, 1174175.
  \bibinfo{pages}{LBNL--5763E, 1174175} pages.
\newblock
\href{https://doi.org/10.2172/1174175}{doi:\nolinkurl{10.2172/1174175}}


\bibitem[{Goldman Sachs}(2025)]%
        {goldman_sachs_ai_2025}
\bibfield{author}{\bibinfo{person}{{Goldman Sachs}}.}
  \bibinfo{year}{2025}\natexlab{}.
\newblock \bibinfo{title}{{AI} to drive 165\% increase in data center power
  demand by 2030}.
\newblock
\urldef\tempurl%
\url{https://www.goldmansachs.com/insights/articles/ai-to-drive-165-increase-in-data-center-power-demand-by-2030}
\showURL{%
\tempurl}


\bibitem[Hall et~al\mbox{.}(2024)]%
        {hall_carbon-aware_2024}
\bibfield{author}{\bibinfo{person}{Sophie Hall}, \bibinfo{person}{Francesco
  Micheli}, \bibinfo{person}{Giuseppe Belgioioso}, \bibinfo{person}{Ana
  Radovanović}, {and} \bibinfo{person}{Florian Dörfler}.}
  \bibinfo{year}{2024}\natexlab{}.
\newblock \bibinfo{title}{Carbon-{Aware} {Computing} for {Data} {Centers} with
  {Probabilistic} {Performance} {Guarantees}}.
\newblock
\href{https://doi.org/10.48550/arXiv.2410.21510}{doi:\nolinkurl{10.48550/arXiv.2410.21510}}
\newblock
\shownote{arXiv:2410.21510 [eess] version: 1}.


\bibitem[Hou et~al\mbox{.}(2023)]%
        {hou_fine-grained_2023}
\bibfield{author}{\bibinfo{person}{Shoulu Hou}, \bibinfo{person}{Wei Ni},
  \bibinfo{person}{Kailan Zhao}, \bibinfo{person}{Bo Cheng},
  \bibinfo{person}{Shuai Zhao}, \bibinfo{person}{Zhiguo Wan},
  \bibinfo{person}{Xiulei Liu}, {and} \bibinfo{person}{Shiping Chen}.}
  \bibinfo{year}{2023}\natexlab{}.
\newblock \showarticletitle{Fine-{Grained} {Online} {Energy} {Management} of
  {Edge} {Data} {Centers} {Using} {Per}-{Core} {Power} {Gating} and {Dynamic}
  {Voltage} and {Frequency} {Scaling}}.
\newblock \bibinfo{journal}{\emph{IEEE Transactions on Sustainable Computing}}
  \bibinfo{volume}{8}, \bibinfo{number}{3} (\bibinfo{date}{July}
  \bibinfo{year}{2023}), \bibinfo{pages}{522--536}.
\newblock
\showISSN{2377-3782}
\href{https://doi.org/10.1109/TSUSC.2023.3250487}{doi:\nolinkurl{10.1109/TSUSC.2023.3250487}}


\bibitem[Howlader and Senjyu(2016)]%
        {howlader_comprehensive_2016}
\bibfield{author}{\bibinfo{person}{Abdul~Motin Howlader} {and}
  \bibinfo{person}{Tomonobu Senjyu}.} \bibinfo{year}{2016}\natexlab{}.
\newblock \showarticletitle{A comprehensive review of low voltage ride through
  capability strategies for the wind energy conversion systems}.
\newblock \bibinfo{journal}{\emph{Renewable and Sustainable Energy Reviews}}
  \bibinfo{volume}{56} (\bibinfo{date}{April} \bibinfo{year}{2016}),
  \bibinfo{pages}{643--658}.
\newblock
\showISSN{13640321}
\href{https://doi.org/10.1016/j.rser.2015.11.073}{doi:\nolinkurl{10.1016/j.rser.2015.11.073}}


\bibitem[Lechl et~al\mbox{.}(2025)]%
        {lechl_uncertainty-aware_2025}
\bibfield{author}{\bibinfo{person}{Michael Lechl}, \bibinfo{person}{Alexander
  Kilian}, {and} \bibinfo{person}{Hermann De~Meer}.}
  \bibinfo{year}{2025}\natexlab{}.
\newblock \showarticletitle{Uncertainty-{Aware} {Scheduling} of {Multi}-{Use}
  {Battery} {Storage} {Systems}}. In \bibinfo{booktitle}{\emph{Proceedings of
  the 16th {ACM} {International} {Conference} on {Future} and {Sustainable}
  {Energy} {Systems}}}. \bibinfo{publisher}{ACM}, \bibinfo{address}{Rotterdam
  Netherlands}, \bibinfo{pages}{243--256}.
\newblock
\showISBNx{979-8-4007-1125-1}
\href{https://doi.org/10.1145/3679240.3734607}{doi:\nolinkurl{10.1145/3679240.3734607}}


\bibitem[Li et~al\mbox{.}(2014)]%
        {li_real-time_2014}
\bibfield{author}{\bibinfo{person}{Na Li}, \bibinfo{person}{Guannan Qu}, {and}
  \bibinfo{person}{Munther Dahleh}.} \bibinfo{year}{2014}\natexlab{}.
\newblock \showarticletitle{Real-time decentralized voltage control in
  distribution networks}. In \bibinfo{booktitle}{\emph{2014 52nd {Annual}
  {Allerton} {Conference} on {Communication}, {Control}, and {Computing}
  ({Allerton})}}. \bibinfo{pages}{582--588}.
\newblock
\href{https://doi.org/10.1109/ALLERTON.2014.7028508}{doi:\nolinkurl{10.1109/ALLERTON.2014.7028508}}


\bibitem[{Liam Ryan} and {Alan Campbell}(2024)]%
        {liam_ryan_shaping_2024}
\bibfield{author}{\bibinfo{person}{{Liam Ryan}} {and} \bibinfo{person}{{Alan
  Campbell}}.} \bibinfo{year}{2024}\natexlab{}.
\newblock \bibinfo{title}{Shaping {Our} {Electricity} {Future} {Advisory}
  {Council}}.
\newblock
\urldef\tempurl%
\url{https://cms.eirgrid.ie/sites/default/files/publications/SOEF-Advisory-Council-Meeting-9-Presentation-24-September-2024.pdf}
\showURL{%
\tempurl}


\bibitem[Lin and Chien(2023)]%
        {lin_adapting_2023}
\bibfield{author}{\bibinfo{person}{Liuzixuan Lin} {and}
  \bibinfo{person}{Andrew~A Chien}.} \bibinfo{year}{2023}\natexlab{}.
\newblock \showarticletitle{Adapting {Datacenter} {Capacity} for {Greener}
  {Datacenters} and {Grid}}. In \bibinfo{booktitle}{\emph{Proceedings of the
  14th {ACM} {International} {Conference} on {Future} {Energy} {Systems}}}.
  \bibinfo{publisher}{ACM}, \bibinfo{address}{Orlando FL USA},
  \bibinfo{pages}{200--213}.
\newblock
\showISBNx{979-8-4007-0032-3}
\href{https://doi.org/10.1145/3575813.3595197}{doi:\nolinkurl{10.1145/3575813.3595197}}


\bibitem[Liu et~al\mbox{.}(2020)]%
        {liu_state---art_2020}
\bibfield{author}{\bibinfo{person}{Lijun Liu}, \bibinfo{person}{Quan Zhang},
  \bibinfo{person}{Zhiqiang~(John) Zhai}, \bibinfo{person}{Chang Yue}, {and}
  \bibinfo{person}{Xiaowei Ma}.} \bibinfo{year}{2020}\natexlab{}.
\newblock \showarticletitle{State-of-the-art on thermal energy storage
  technologies in data center}.
\newblock \bibinfo{journal}{\emph{Energy and Buildings}}  \bibinfo{volume}{226}
  (\bibinfo{date}{Nov.} \bibinfo{year}{2020}), \bibinfo{pages}{110345}.
\newblock
\showISSN{03787788}
\href{https://doi.org/10.1016/j.enbuild.2020.110345}{doi:\nolinkurl{10.1016/j.enbuild.2020.110345}}


\bibitem[Liu et~al\mbox{.}(2013)]%
        {liu_data_2013}
\bibfield{author}{\bibinfo{person}{Zhenhua Liu}, \bibinfo{person}{Adam
  Wierman}, \bibinfo{person}{Yuan Chen}, \bibinfo{person}{Benjamin Razon},
  {and} \bibinfo{person}{Niangjun Chen}.} \bibinfo{year}{2013}\natexlab{}.
\newblock \showarticletitle{Data center demand response: avoiding the
  coincident peak via workload shifting and local generation}. In
  \bibinfo{booktitle}{\emph{Proceedings of the {ACM} {SIGMETRICS}/international
  conference on {Measurement} and modeling of computer systems}}.
  \bibinfo{publisher}{ACM}, \bibinfo{address}{Pittsburgh PA USA},
  \bibinfo{pages}{341--342}.
\newblock
\showISBNx{978-1-4503-1900-3}
\href{https://doi.org/10.1145/2465529.2465740}{doi:\nolinkurl{10.1145/2465529.2465740}}


\bibitem[Moheb et~al\mbox{.}(2022)]%
        {moheb_comprehensive_2022}
\bibfield{author}{\bibinfo{person}{Aya~M. Moheb}, \bibinfo{person}{Enas~A.
  El-Hay}, {and} \bibinfo{person}{Attia~A. El-Fergany}.}
  \bibinfo{year}{2022}\natexlab{}.
\newblock \showarticletitle{Comprehensive {Review} on {Fault} {Ride}-{Through}
  {Requirements} of {Renewable} {Hybrid} {Microgrids}}.
\newblock \bibinfo{journal}{\emph{Energies}} \bibinfo{volume}{15},
  \bibinfo{number}{18} (\bibinfo{date}{Sept.} \bibinfo{year}{2022}),
  \bibinfo{pages}{6785}.
\newblock
\showISSN{1996-1073}
\href{https://doi.org/10.3390/en15186785}{doi:\nolinkurl{10.3390/en15186785}}


\bibitem[MYEN(2024)]%
        {myen_technical_2024}
\bibfield{author}{\bibinfo{person}{MYEN}.} \bibinfo{year}{2024}\natexlab{}.
\newblock \bibinfo{title}{Technical {Regulation} 3.4.3 {Requirements} for
  {Transmission}-{Connected} {Demand} {Facilities}, {Revision} 1}.
\newblock
\urldef\tempurl%
\url{https://en.energinet.dk/media/bkpf4hef/technical-regulation-343-requirements-for-transmission-connected-demand-facilities-revision-1.pdf}
\showURL{%
\tempurl}


\bibitem[NERC(2025a)]%
        {nerc_2024_2025}
\bibfield{author}{\bibinfo{person}{NERC}.} \bibinfo{year}{2025}\natexlab{a}.
\newblock \bibinfo{booktitle}{\emph{2024 {Long}-{Term} {Reliability}
  {Assessment} {Report}}}.
\newblock \bibinfo{type}{{T}echnical {R}eport}. \bibinfo{institution}{North
  American Electric Reliability Corporation}.
\newblock
\urldef\tempurl%
\url{https://www.nerc.com/pa/RAPA/ra/Reliability%20Assessments%20DL/NERC_Long%20Term%20Reliability%20Assessment_2024.pdf}
\showURL{%
\tempurl}


\bibitem[NERC(2025b)]%
        {nerc_nerc_2025}
\bibfield{author}{\bibinfo{person}{NERC}.} \bibinfo{year}{2025}\natexlab{b}.
\newblock \bibinfo{booktitle}{\emph{{NERC} {Incident} {Review} {Considering}
  {Simultaneous} {Voltage}-{Sensitive} {Load} {Reductions}}}.
\newblock \bibinfo{type}{{T}echnical {R}eport}.
\newblock
\urldef\tempurl%
\url{https://www.nerc.com/pa/rrm/ea/Documents/Incident_Review_Large_Load_Loss.pdf}
\showURL{%
\tempurl}


\bibitem[Paananen and Nasr(2021)]%
        {paananen_grid-interactive_2021}
\bibfield{author}{\bibinfo{person}{Janne Paananen} {and} \bibinfo{person}{Ehsan
  Nasr}.} \bibinfo{year}{2021}\natexlab{}.
\newblock \showarticletitle{Grid-interactive data centers: enabling
  decarbonization and system stability}.
\newblock  (\bibinfo{date}{July} \bibinfo{year}{2021}).
\newblock


\bibitem[Patel et~al\mbox{.}(2024)]%
        {patel_characterizing_2024}
\bibfield{author}{\bibinfo{person}{Pratyush Patel}, \bibinfo{person}{Esha
  Choukse}, \bibinfo{person}{Chaojie Zhang}, \bibinfo{person}{Íñigo Goiri},
  \bibinfo{person}{Brijesh Warrier}, \bibinfo{person}{Nithish Mahalingam},
  {and} \bibinfo{person}{Ricardo Bianchini}.} \bibinfo{year}{2024}\natexlab{}.
\newblock \showarticletitle{Characterizing {Power} {Management} {Opportunities}
  for {LLMs} in the {Cloud}}. In \bibinfo{booktitle}{\emph{Proceedings of the
  29th {ACM} {International} {Conference} on {Architectural} {Support} for
  {Programming} {Languages} and {Operating} {Systems}, {Volume} 3}}.
  \bibinfo{publisher}{ACM}, \bibinfo{address}{La Jolla CA USA},
  \bibinfo{pages}{207--222}.
\newblock
\showISBNx{979-8-4007-0386-7}
\href{https://doi.org/10.1145/3620666.3651329}{doi:\nolinkurl{10.1145/3620666.3651329}}


\bibitem[{Patrick Gravois}(2025)]%
        {patrick_gravois_ercot_2025}
\bibfield{author}{\bibinfo{person}{{Patrick Gravois}}.}
  \bibinfo{year}{2025}\natexlab{}.
\newblock \bibinfo{title}{{ERCOT} {Large} {Electronic} {Load} {Voltage}
  {Ride}-{Through} {Performance} {Requirements} ({Proposal})}.
\newblock
\urldef\tempurl%
\url{https://www.ercot.com/files/docs/2025/07/11/ERCOT-LEL-Ride-Through-Criteria_LLWG-final.pptx}
\showURL{%
\tempurl}


\bibitem[Priya et~al\mbox{.}(2024)]%
        {priya_energy-minimizing_2024}
\bibfield{author}{\bibinfo{person}{Aditya Priya}, \bibinfo{person}{Rajiv
  Choudhury}, \bibinfo{person}{Sujay Patni}, \bibinfo{person}{Himkant Sharma},
  \bibinfo{person}{Moonmoon Mohanty}, \bibinfo{person}{Krishnasuri Narayanam},
  \bibinfo{person}{Umamaheswari Devi}, \bibinfo{person}{Pratibha Moogi},
  \bibinfo{person}{Preetam Patil}, {and} \bibinfo{person}{Parimal Parag}.}
  \bibinfo{year}{2024}\natexlab{}.
\newblock \showarticletitle{Energy-minimizing workload splitting and frequency
  selection for guaranteed performance over heterogeneous cores}. In
  \bibinfo{booktitle}{\emph{The 15th {ACM} {International} {Conference} on
  {Future} and {Sustainable} {Energy} {Systems}}}. \bibinfo{publisher}{ACM},
  \bibinfo{address}{Singapore Singapore}, \bibinfo{pages}{308--322}.
\newblock
\showISBNx{979-8-4007-0480-2}
\href{https://doi.org/10.1145/3632775.3661968}{doi:\nolinkurl{10.1145/3632775.3661968}}


\bibitem[Qu and Li(2020)]%
        {qu_optimal_2020}
\bibfield{author}{\bibinfo{person}{Guannan Qu} {and} \bibinfo{person}{Na Li}.}
  \bibinfo{year}{2020}\natexlab{}.
\newblock \showarticletitle{Optimal {Distributed} {Feedback} {Voltage}
  {Control} {Under} {Limited} {Reactive} {Power}}.
\newblock \bibinfo{journal}{\emph{IEEE Transactions on Power Systems}}
  \bibinfo{volume}{35}, \bibinfo{number}{1} (\bibinfo{date}{Jan.}
  \bibinfo{year}{2020}), \bibinfo{pages}{315--331}.
\newblock
\showISSN{1558-0679}
\href{https://doi.org/10.1109/TPWRS.2019.2931685}{doi:\nolinkurl{10.1109/TPWRS.2019.2931685}}


\bibitem[Radovanovic et~al\mbox{.}(2021)]%
        {radovanovic_carbon-aware_2021}
\bibfield{author}{\bibinfo{person}{Ana Radovanovic}, \bibinfo{person}{Ross
  Koningstein}, \bibinfo{person}{Ian Schneider}, \bibinfo{person}{Bokan Chen},
  \bibinfo{person}{Alexandre Duarte}, \bibinfo{person}{Binz Roy},
  \bibinfo{person}{Diyue Xiao}, \bibinfo{person}{Maya Haridasan},
  \bibinfo{person}{Patrick Hung}, \bibinfo{person}{Nick Care},
  \bibinfo{person}{Saurav Talukdar}, \bibinfo{person}{Eric Mullen},
  \bibinfo{person}{Kendal Smith}, \bibinfo{person}{MariEllen Cottman}, {and}
  \bibinfo{person}{Walfredo Cirne}.} \bibinfo{year}{2021}\natexlab{}.
\newblock \bibinfo{title}{Carbon-{Aware} {Computing} for {Datacenters}}.
\newblock
\href{https://doi.org/10.48550/arXiv.2106.11750}{doi:\nolinkurl{10.48550/arXiv.2106.11750}}
\newblock
\shownote{arXiv:2106.11750 [cs]}.


\bibitem[Rong et~al\mbox{.}(2016)]%
        {rong_optimizing_2016}
\bibfield{author}{\bibinfo{person}{Huigui Rong}, \bibinfo{person}{Haomin
  Zhang}, \bibinfo{person}{Sheng Xiao}, \bibinfo{person}{Canbing Li}, {and}
  \bibinfo{person}{Chunhua Hu}.} \bibinfo{year}{2016}\natexlab{}.
\newblock \showarticletitle{Optimizing energy consumption for data centers}.
\newblock \bibinfo{journal}{\emph{Renewable and Sustainable Energy Reviews}}
  \bibinfo{volume}{58} (\bibinfo{date}{May} \bibinfo{year}{2016}),
  \bibinfo{pages}{674--691}.
\newblock
\showISSN{13640321}
\href{https://doi.org/10.1016/j.rser.2015.12.283}{doi:\nolinkurl{10.1016/j.rser.2015.12.283}}


\bibitem[{RTE}(2024)]%
        {rte_article_2024}
\bibfield{author}{\bibinfo{person}{{RTE}}.} \bibinfo{year}{2024}\natexlab{}.
\newblock \bibinfo{title}{Article 8.3.5 – {Cahier} des {Charges} des
  capacités constructives d’une {Installation} de consommation}.
\newblock
\urldef\tempurl%
\url{https://www.services-rte.com/files/live//sites/services-rte/files/documentsLibrary/Article_8.3.5__CdC_des_capacites_constructive_d_une_installation_de_consommateurs_1897_fr}
\showURL{%
\tempurl}


\bibitem[Savasci et~al\mbox{.}(2024)]%
        {savasci_pads_2024}
\bibfield{author}{\bibinfo{person}{Mehmet Savasci}, \bibinfo{person}{Abel
  Souza}, \bibinfo{person}{David Irwin}, \bibinfo{person}{Ahmed Ali-Eldin},
  {and} \bibinfo{person}{Prashant Shenoy}.} \bibinfo{year}{2024}\natexlab{}.
\newblock \showarticletitle{{PADS}: {Power} {Budgeting} with {Diagonal}
  {Scaling} for {Performance}-{Aware} {Cloud} {Workloads}}. In
  \bibinfo{booktitle}{\emph{2024 {IEEE} 15th {International} {Green} and
  {Sustainable} {Computing} {Conference} ({IGSC})}}. \bibinfo{pages}{14--21}.
\newblock
\href{https://doi.org/10.1109/IGSC64514.2024.00012}{doi:\nolinkurl{10.1109/IGSC64514.2024.00012}}
\newblock
\shownote{ISSN: 2993-2084}.


\bibitem[Sawyer({[n.\,d.]})]%
        {sawyer_calculating_nodate}
\bibfield{author}{\bibinfo{person}{Richard~L Sawyer}.}
  \bibinfo{year}{[n.\,d.]}\natexlab{}.
\newblock \showarticletitle{Calculating {Total} {Power} {Requirements} for
  {Data} {Center}}.
\newblock  (\bibinfo{year}{[n.\,d.]}).
\newblock


\bibitem[Shehabi et~al\mbox{.}(2024)]%
        {shehabi_united_2024}
\bibfield{author}{\bibinfo{person}{Arman Shehabi}, \bibinfo{person}{Sarah
  Smith}, \bibinfo{person}{Dale Sartor}, \bibinfo{person}{Richard Brown},
  \bibinfo{person}{Magnus Herrlin}, \bibinfo{person}{Jonathan Koomey},
  \bibinfo{person}{Eric Masanet}, \bibinfo{person}{Nathaniel Horner},
  \bibinfo{person}{Inês Azevedo}, {and} \bibinfo{person}{William Lintner}.}
  \bibinfo{year}{2024}\natexlab{}.
\newblock \bibinfo{booktitle}{\emph{United {States} {Data} {Center} {Energy}
  {Usage} {Report}}}.
\newblock \bibinfo{type}{{T}echnical {R}eport} LBNL--1005775, 1372902.
  \bibinfo{pages}{LBNL--1005775, 1372902} pages.
\newblock
\href{https://doi.org/10.2172/1372902}{doi:\nolinkurl{10.2172/1372902}}


\bibitem[Turitsyn et~al\mbox{.}(2011)]%
        {turitsyn_options_2011}
\bibfield{author}{\bibinfo{person}{Konstantin Turitsyn}, \bibinfo{person}{Petr
  Sulc}, \bibinfo{person}{Scott Backhaus}, {and} \bibinfo{person}{Michael
  Chertkov}.} \bibinfo{year}{2011}\natexlab{}.
\newblock \showarticletitle{Options for {Control} of {Reactive} {Power} by
  {Distributed} {Photovoltaic} {Generators}}.
\newblock \bibinfo{journal}{\emph{Proc. IEEE}} \bibinfo{volume}{99},
  \bibinfo{number}{6} (\bibinfo{date}{June} \bibinfo{year}{2011}),
  \bibinfo{pages}{1063--1073}.
\newblock
\showISSN{1558-2256}
\href{https://doi.org/10.1109/JPROC.2011.2116750}{doi:\nolinkurl{10.1109/JPROC.2011.2116750}}


\bibitem[Union(2016)]%
        {european_union_commission_2016}
\bibfield{author}{\bibinfo{person}{European Union}.}
  \bibinfo{year}{2016}\natexlab{}.
\newblock \bibinfo{title}{Commission {Regulation} ({EU}) 2016/1388}.
\newblock
\urldef\tempurl%
\url{https://eur-lex.europa.eu/eli/reg/2016/1388/oj/eng}
\showURL{%
\tempurl}


\bibitem[Vaidhynathan et~al\mbox{.}(2025)]%
        {vaidhynathan_vulcan_2025}
\bibfield{author}{\bibinfo{person}{Deepthi Vaidhynathan},
  \bibinfo{person}{Kumaraguru Prabakar}, \bibinfo{person}{Gregory Martin},
  \bibinfo{person}{Anand Ramesh}, \bibinfo{person}{Ben Wheeler},
  \bibinfo{person}{Christopher Coco}, \bibinfo{person}{Jimmy Clidaras},
  \bibinfo{person}{Matthieu Monsch}, \bibinfo{person}{Sangsun Kim},
  \bibinfo{person}{Saurav Talukdar}, {and} \bibinfo{person}{Shilpa Marti}.}
  \bibinfo{year}{2025}\natexlab{}.
\newblock \showarticletitle{Vulcan {Test} {Platform}: {Demonstrating} the
  {Data} {Center} as a {Flexible} {Grid} {Asset}}.
\newblock \bibinfo{journal}{\emph{Renewable Energy}} (\bibinfo{year}{2025}).
\newblock


\bibitem[Wang et~al\mbox{.}(2019)]%
        {wang_frequency_2019}
\bibfield{author}{\bibinfo{person}{Wei Wang}, \bibinfo{person}{Amirali
  Abdolrashidi}, \bibinfo{person}{Nanpeng Yu}, {and} \bibinfo{person}{Daniel
  Wong}.} \bibinfo{year}{2019}\natexlab{}.
\newblock \showarticletitle{Frequency regulation service provision in data
  center with computational flexibility}.
\newblock \bibinfo{journal}{\emph{Applied Energy}}  \bibinfo{volume}{251}
  (\bibinfo{date}{Oct.} \bibinfo{year}{2019}), \bibinfo{pages}{113304}.
\newblock
\showISSN{03062619}
\href{https://doi.org/10.1016/j.apenergy.2019.05.107}{doi:\nolinkurl{10.1016/j.apenergy.2019.05.107}}


\bibitem[Watson(2025)]%
        {watson_data_2025}
\bibfield{author}{\bibinfo{person}{Keith Watson}.}
  \bibinfo{year}{2025}\natexlab{}.
\newblock \showarticletitle{Data {Centers} – {A} {Good} {Grid} {Citizen}}.
\newblock  (\bibinfo{date}{July} \bibinfo{year}{2025}).
\newblock


\bibitem[Wierman et~al\mbox{.}(2014)]%
        {wierman_opportunities_2014}
\bibfield{author}{\bibinfo{person}{Adam Wierman}, \bibinfo{person}{Zhenhua
  Liu}, \bibinfo{person}{Iris Liu}, {and} \bibinfo{person}{Hamed
  Mohsenian-Rad}.} \bibinfo{year}{2014}\natexlab{}.
\newblock \showarticletitle{Opportunities and challenges for data center demand
  response}. In \bibinfo{booktitle}{\emph{International {Green} {Computing}
  {Conference}}}. \bibinfo{pages}{1--10}.
\newblock
\href{https://doi.org/10.1109/IGCC.2014.7039172}{doi:\nolinkurl{10.1109/IGCC.2014.7039172}}


\bibitem[Wu et~al\mbox{.}(2018)]%
        {wu_smart_2018}
\bibfield{author}{\bibinfo{person}{Chenye Wu}, \bibinfo{person}{Gabriela Hug},
  {and} \bibinfo{person}{Soummya Kar}.} \bibinfo{year}{2018}\natexlab{}.
\newblock \showarticletitle{Smart {Inverter} for {Voltage} {Regulation}:
  {Physical} and {Market} {Implementation}}.
\newblock \bibinfo{journal}{\emph{IEEE Transactions on Power Systems}}
  \bibinfo{volume}{33}, \bibinfo{number}{6} (\bibinfo{date}{Nov.}
  \bibinfo{year}{2018}), \bibinfo{pages}{6181--6192}.
\newblock
\showISSN{1558-0679}
\href{https://doi.org/10.1109/TPWRS.2018.2854903}{doi:\nolinkurl{10.1109/TPWRS.2018.2854903}}


\bibitem[Xia et~al\mbox{.}(2019)]%
        {xia_galerkin_2019}
\bibfield{author}{\bibinfo{person}{Bingqing Xia}, \bibinfo{person}{Hao Wu},
  \bibinfo{person}{Yiwei Qiu}, \bibinfo{person}{Boliang Lou}, {and}
  \bibinfo{person}{Yonghua Song}.} \bibinfo{year}{2019}\natexlab{}.
\newblock \showarticletitle{A {Galerkin} {Method}-{Based} {Polynomial}
  {Approximation} for {Parametric} {Problems} in {Power} {System} {Transient}
  {Analysis}}.
\newblock \bibinfo{journal}{\emph{IEEE Transactions on Power Systems}}
  \bibinfo{volume}{34}, \bibinfo{number}{2} (\bibinfo{date}{March}
  \bibinfo{year}{2019}), \bibinfo{pages}{1620--1629}.
\newblock
\showISSN{1558-0679}
\href{https://doi.org/10.1109/TPWRS.2018.2879367}{doi:\nolinkurl{10.1109/TPWRS.2018.2879367}}


\bibitem[Yeh et~al\mbox{.}(2022)]%
        {yeh_robust_2022}
\bibfield{author}{\bibinfo{person}{Christopher Yeh}, \bibinfo{person}{Jing Yu},
  \bibinfo{person}{Yuanyuan Shi}, {and} \bibinfo{person}{Adam Wierman}.}
  \bibinfo{year}{2022}\natexlab{}.
\newblock \showarticletitle{Robust online voltage control with an unknown grid
  topology}. In \bibinfo{booktitle}{\emph{Proceedings of the {Thirteenth} {ACM}
  {International} {Conference} on {Future} {Energy} {Systems}}}
  \emph{(\bibinfo{series}{e-{Energy} '22})}. \bibinfo{publisher}{Association
  for Computing Machinery}, \bibinfo{address}{New York, NY, USA},
  \bibinfo{pages}{240--250}.
\newblock
\showISBNx{9781450393973}
\href{https://doi.org/10.1145/3538637.3538853}{doi:\nolinkurl{10.1145/3538637.3538853}}


\bibitem[Zeb et~al\mbox{.}(2022)]%
        {zeb_faults_2022}
\bibfield{author}{\bibinfo{person}{Kamran Zeb}, \bibinfo{person}{Saif~Ul
  Islam}, \bibinfo{person}{Imran Khan}, \bibinfo{person}{Waqar Uddin},
  \bibinfo{person}{M. Ishfaq}, \bibinfo{person}{Tiago~Davi Curi~Busarello},
  \bibinfo{person}{S.M. Muyeen}, \bibinfo{person}{Iftikhar Ahmad}, {and}
  \bibinfo{person}{H.J. Kim}.} \bibinfo{year}{2022}\natexlab{}.
\newblock \showarticletitle{Faults and {Fault} {Ride} {Through} strategies for
  grid-connected photovoltaic system: {A} comprehensive review}.
\newblock \bibinfo{journal}{\emph{Renewable and Sustainable Energy Reviews}}
  \bibinfo{volume}{158} (\bibinfo{date}{April} \bibinfo{year}{2022}),
  \bibinfo{pages}{112125}.
\newblock
\showISSN{13640321}
\href{https://doi.org/10.1016/j.rser.2022.112125}{doi:\nolinkurl{10.1016/j.rser.2022.112125}}


\bibitem[Zhang et~al\mbox{.}(2013)]%
        {zhang_local_2013}
\bibfield{author}{\bibinfo{person}{Baosen Zhang}, \bibinfo{person}{Alejandro~D.
  Domínguez-García}, {and} \bibinfo{person}{David Tse}.}
  \bibinfo{year}{2013}\natexlab{}.
\newblock \showarticletitle{A local control approach to voltage regulation in
  distribution networks}. In \bibinfo{booktitle}{\emph{2013 {North} {American}
  {Power} {Symposium} ({NAPS})}}. \bibinfo{pages}{1--6}.
\newblock
\href{https://doi.org/10.1109/NAPS.2013.6666945}{doi:\nolinkurl{10.1109/NAPS.2013.6666945}}


\bibitem[Zhang et~al\mbox{.}(2015)]%
        {zhang_optimal_2015}
\bibfield{author}{\bibinfo{person}{Baosen Zhang}, \bibinfo{person}{Albert~Y.S.
  Lam}, \bibinfo{person}{Alejandro~D. Domínguez-García}, {and}
  \bibinfo{person}{David Tse}.} \bibinfo{year}{2015}\natexlab{}.
\newblock \showarticletitle{An {Optimal} and {Distributed} {Method} for
  {Voltage} {Regulation} in {Power} {Distribution} {Systems}}.
\newblock \bibinfo{journal}{\emph{IEEE Transactions on Power Systems}}
  \bibinfo{volume}{30}, \bibinfo{number}{4} (\bibinfo{date}{July}
  \bibinfo{year}{2015}), \bibinfo{pages}{1714--1726}.
\newblock
\showISSN{1558-0679}
\href{https://doi.org/10.1109/TPWRS.2014.2347281}{doi:\nolinkurl{10.1109/TPWRS.2014.2347281}}


\bibitem[Zhang et~al\mbox{.}(2021)]%
        {zhang_survey_2021}
\bibfield{author}{\bibinfo{person}{Qingxia Zhang}, \bibinfo{person}{Zihao
  Meng}, \bibinfo{person}{Xianwen Hong}, \bibinfo{person}{Yuhao Zhan},
  \bibinfo{person}{Jia Liu}, \bibinfo{person}{Jiabao Dong},
  \bibinfo{person}{Tian Bai}, \bibinfo{person}{Junyu Niu}, {and}
  \bibinfo{person}{M.~Jamal Deen}.} \bibinfo{year}{2021}\natexlab{}.
\newblock \showarticletitle{A survey on data center cooling systems:
  {Technology}, power consumption modeling and control strategy optimization}.
\newblock \bibinfo{journal}{\emph{Journal of Systems Architecture}}
  \bibinfo{volume}{119} (\bibinfo{date}{Oct.} \bibinfo{year}{2021}),
  \bibinfo{pages}{102253}.
\newblock
\showISSN{13837621}
\href{https://doi.org/10.1016/j.sysarc.2021.102253}{doi:\nolinkurl{10.1016/j.sysarc.2021.102253}}


\bibitem[Zhu and Liu(2016)]%
        {zhu_fast_2016}
\bibfield{author}{\bibinfo{person}{Hao Zhu} {and} \bibinfo{person}{Hao~Jan
  Liu}.} \bibinfo{year}{2016}\natexlab{}.
\newblock \showarticletitle{Fast {Local} {Voltage} {Control} {Under} {Limited}
  {Reactive} {Power}: {Optimality} and {Stability} {Analysis}}.
\newblock \bibinfo{journal}{\emph{IEEE Transactions on Power Systems}}
  \bibinfo{volume}{31}, \bibinfo{number}{5} (\bibinfo{date}{Sept.}
  \bibinfo{year}{2016}), \bibinfo{pages}{3794--3803}.
\newblock
\showISSN{1558-0679}
\href{https://doi.org/10.1109/TPWRS.2015.2504419}{doi:\nolinkurl{10.1109/TPWRS.2015.2504419}}


\bibitem[Zhu and Wang(2025)]%
        {zhu_energy-efficient_2025}
\bibfield{author}{\bibinfo{person}{Shuntao Zhu} {and} \bibinfo{person}{Dan
  Wang}.} \bibinfo{year}{2025}\natexlab{}.
\newblock \showarticletitle{Energy-efficient {LLM} {Training} in {GPU}
  datacenters with {Immersion} {Cooling} {Systems}}. In
  \bibinfo{booktitle}{\emph{Proceedings of the 16th {ACM} {International}
  {Conference} on {Future} and {Sustainable} {Energy} {Systems}}}.
  \bibinfo{publisher}{ACM}, \bibinfo{address}{Rotterdam Netherlands},
  \bibinfo{pages}{407--414}.
\newblock
\showISBNx{979-8-4007-1125-1}
\href{https://doi.org/10.1145/3679240.3734609}{doi:\nolinkurl{10.1145/3679240.3734609}}


\end{thebibliography}

\appendix









\end{document}